\DocumentMetadata{lang=en}
\documentclass[a4paper]{article}

\usepackage[utf8]{inputenc}

\usepackage[hidelinks,hypertexnames=false]{hyperref}
\usepackage{graphicx}

\usepackage[backend=bibtex,citestyle=numeric-comp,sorting=nyt,maxnames=4,backref]{biblatex}
\DefineBibliographyStrings{english}{
    backrefpage = {$\hookrightarrow$~p.},
    backrefpages = {$\hookrightarrow$~pp.}
}
\usepackage{mathtools}
\usepackage{amsthm}
\usepackage{amsfonts}
\usepackage{amssymb}
\usepackage{stmaryrd}
\numberwithin{equation}{section}
\numberwithin{figure}{section}

\usepackage{thmtools}
\usepackage{zref-clever}
\zcsetup{cap}
\newcommand{\Cref}{\zcref}

\declaretheorem[numberwithin=section]{theorem}
\declaretheorem[numbered=no]{conjecture}
\declaretheorem[sibling=theorem]{proposition}
\declaretheorem[sibling=theorem]{corollary}
\declaretheorem[sibling=theorem]{lemma}
\declaretheorem[sibling=theorem, style=definition]{remark}
\declaretheorem[sibling=theorem, style=definition]{definition}

\renewcommand\epsilon\varepsilon
\newcommand{\abs}[1]{\left|#1\right|}
\newcommand{\norm}[1]{\left\|#1\right\|}
\newcommand{\smallnorm}[1]{\|#1\|}
\newcommand{\I}{\mathbf 1}
\newcommand{\N}{\mathbb N}
\newcommand{\R}{\mathbb R}
\newcommand{\Z}{\mathbb Z}
\newcommand{\Laplace}{\Delta}
\newcommand{\Tor}{\mathbb T}
\newcommand{\Holder}{\mathcal C}
\newcommand\diff{\mathop{}\!\mathrm d}
\newcommand\dt{\diff t}
\newcommand\du{\diff u}
\newcommand\dx{\diff x}
\newcommand\dy{\diff y}
\newcommand\dz{\diff z}
\newcommand\bigO{\mathcal O}
\newcommand\pospart[1]{\llbracket#1\rrbracket_+}
\DeclareMathOperator{\capa}{cap}
\DeclareMathOperator{\E}{\mathbb E}
\DeclareMathOperator{\Hess}{Hess}
\DeclareMathOperator{\Prob}{\mathbb P}
\DeclareMathOperator{\projN}{P_{\!\mathit N}}
\DeclareMathOperator{\projM}{P_{\!\mathit M}}
\DeclareMathOperator{\projMperp}{P_{\!\mathit M}^\perp}
\DeclareMathOperator{\projNperp}{P_{\!\mathit N}^\perp}
\DeclareMathOperator{\projperp}{P_{\!\perp}}
\DeclareMathOperator{\cd}{cd}
\DeclareMathOperator{\cova}{\mathcal L}
\DeclareMathOperator{\dist}{d}
\renewcommand{\complement}{c}

\newcommand{\manif}{\mathcal M}
\newcommand{\neigh}{\mathcal D}
\newcommand{\Blog}{B_{\log}}

\title{Metastable transition times of the\\ 1D dynamical sine-Gordon model}
\author{Petri Laarne\footnote{
    Department of Mathematics and Statistics, University of Helsinki, Finland;
    \newline\hspace*{1.8em}\url{petri.laarne@helsinki.fi}}}
\date{}

\begin{document}

\maketitle
\begin{abstract}
We study the dynamics of a stochastic heat equation with $\gamma\sin(\beta u)$ nonlinearity
on one-dimensional torus.
We show an Eyring--Kramers law for the jump rate between potential wells
in the small-noise limit,
and that the transition state undergoes a bifurcation at $\gamma\beta = 1$.
The argument follows the potential-theoretic approach of Berglund and Gentz [Electron.\ J.\ Probab.\ 2013].

    (Secondary)

\end{abstract}

\section{Introduction}

We study the \emph{dynamical sine-Gordon equation}
\begin{equation}\label{eq:DSG}
\partial_t u - \Laplace u = -\gamma \sin(\beta u) + \sqrt{2\epsilon} \xi,
\end{equation}
where $\gamma, \beta > 0$ are fixed parameters,
$\xi$ is space-time white noise,
and the strength of the noise $\epsilon$ is taken to be very small.
The spatial domain is the one-dimensional torus $\Tor$, which we fix to have length $2\pi$.

When $\epsilon = 0$, the equation has stable stationary points
at constant solutions $u \equiv 2\pi k / \beta$, $k \in \Z$.
The stochastic forcing turns these potential minima \emph{metastable}:
a solution will spend a long time near $2\pi k / \beta$,
but will occasionally jump to a neighbourhood of $2\pi (k \pm 1) / \beta$.

We will answer the following two questions in this article.

\bigskip\noindent\textbf{Question 1.}
How often, on average, does a solution to \eqref{eq:DSG}
jump from the neighbourhood of one minimum $2\pi k/\beta$ to another?

\bigskip\noindent\textbf{Question 2.}
What does the solution look like in the middle of a transition?

\bigskip
The first question goes back to the works of Eyring and Kramers on chemical reaction rates
\cite{eyring_activated_1935,kramers_brownian_1940};
see the survey \cite{hanggi_reactionrate_1990} for more on the history.
The so-called \emph{Eyring--Kramers law} states that the transition rate
is proportional to $\exp(-h/\epsilon)$,
where $h$ is the depth of the potential well.

There have been several mathematical approaches to proving the Eyring--Kramers law,
e.g.\ the large deviations theory of Freĭdlin and Wentzell \cite{freidlin_random_1984}.
We discuss it in \Cref{sec:ldp}.
Our main argument however builds on the potential-theoretic approach of Bovier, Eckhoff, Gayrard, and Klein
\cite{bovier_metastability_2004,bovier_metastability_2005},
which provides a sharp subexponential prefactor.
We outline it in \Cref{sec:potential theory}.

The potential-theoretic method assumes a finite-dimensional diffusion,
but an approximation argument for stochastic PDEs
was developed by Berglund and Gentz \cite{berglund_sharp_2013}.
There they study a double-well model given by a fourth-order potential,
with either Neumann or periodic boundary conditions.
As the cosine potential underlying \eqref{eq:DSG} has similar locally quadratic behaviour,
the present argument is a minor modification to theirs.

In other lines of development, the recent article \cite{avelin_geometric_2023}
rephrases the potential-theoretic method in terms of function-theoretic quantities,
but the results are not directly applicable to our case.
There are also spectral approaches to metastability,
as the expected transition time is related to the spectral gap of the generator of the diffusion
(see e.g.\ \cite{bovier_metastability_2005}).

\bigskip\noindent
The second question can be answered with tools from Hamiltonian mechanics
and the deterministic version of \eqref{eq:DSG}.
As the likelihood of a transition is exponentially related to the potential barrier height,
a typical transition passes by the \emph{unstable} stationary solution of minimal energy.
Such solutions are called \emph{transition states} or \emph{instantons}.

The result derived in \Cref{sec:stationary} can be summarized as:
\begin{itemize}
\item When $\gamma\beta < 1$, the transition state between $2k\pi/\beta$ and $2(k+1)\pi/\beta$
    is the constant function $(2k+1)\pi/\beta$.
\item When $\gamma\beta > 1$, the transition states are kink--antikink pairs
    given by an explicit formula, unique up to translation in $x$.
\end{itemize}

The sine-Gordon equation becomes much more interesting in two spatial dimensions.
There the space-time white noise is so irregular that equation~\eqref{eq:DSG}
does not make sense as is.
The nonlinearity must be renormalized, and the relevant equation is instead the $\delta \to 0$ limit of
\begin{equation}\label{eq:DSG 2D}
\partial_t u - \Laplace u = -\gamma_\delta \sin(\beta u) + \sqrt{2\epsilon} \xi_\delta,
\end{equation}
where $\xi_\delta$ is a suitably mollified noise
and $\gamma_\delta \propto \delta^{-C\epsilon\beta^2}$ as $\delta \to 0$
(see e.g.\ \cite[Theorem~1.1]{hairer_dynamical_2016} for precise definitions).
This equation arises in mathematical physics (e.g.\ in relation to Coulomb gas),
and is an interesting test case for renormalization techniques due to the non-polynomial nonlinearity.
We refer to \cite{bringmann_global_2024,chandra_dynamical_2018,hairer_dynamical_2016,gubinelli_simple_2024}
for more background.

A renormalized double-well model, with nonlinearity of form $-u^3 + \epsilon C_\delta u$,
was studied by Berglund, Di~Gesù, and Weber \cite{berglund_eyring_2017};
see also \cite{barashkov_eyringkramers_2024}.
There the divergent $C_\delta u \to \infty u$ term compensates for the vanishing product
in the prefactor~\eqref{eq:result easy prefactor}.
Even though the potential wells are formally infinitely far apart,
the dynamics of the equation still matches the deterministic model.

In our case the renormalization is multiplicative and not additive.
It is an interesting open question whether the renormalized model can still be analyzed in a similar way.
In particular the dependence on $\beta$ would be interesting to understand,
as the 2D wellposedness theory also features $\beta$-thresholds.

The literature on the 2D double-well model is currently
restricted to the case where a constant function is the transition state.
The non-constant case is also an interesting open problem.

\bigskip\noindent
In this article we adapt the capacity computation of Berglund and Gentz
to the 1D dynamical sine-Gordon equation.
Especially in the $\gamma\beta > 1$ case we aim for a self-contained, detailed, and pedagogic presentation.
The article~\cite{berglund_sharp_2013} mostly focuses on Neumann boundary conditions
and skips many details on the periodic case,
which however has some additional nontrivial steps.

We always consider mild solutions to the equation;
their existence and uniqueness is stated in \Cref{thm:wellposedness}.
Furthermore, we pass to the spectral Galerkin approximations
\begin{equation}\label{eq:DSG truncated}
\partial_t u - \Laplace u = -\gamma \projN \sin(\beta \projN u) + \sqrt{2\epsilon} \projN \xi,
\end{equation}
where $\projN$ is projection to wavenumbers of magnitude at most $N$.
We use the potential-theoretic framework to then find uniform-in-$N$ estimates
for the transition times.

Since the potential is $2\pi/\beta$-periodic,
we can take the initial data to always be in a neighbourhood of $0$,
and translate the process once it has entered an adjacent potential well.
We also add a confining term to the potential energy:
\begin{equation}\label{eq:confining potential}
F[u] \coloneqq \int_\Tor \frac{\abs{\nabla u(x)}^2}{2}
        - \frac\gamma\beta \cos(\beta u(x)) \dx
    + \max(0, |\hat u(0)| - k2\pi\beta)^2,
\end{equation}
where $k \in \Z_+$ is arbitrary but finite.

\begin{remark}
The confining term is a technical assumption used in
\Cref{thm:sg mountain pass,thm:hard capacity stable manifold}.
The choice of $k$ in~\eqref{eq:confining potential}
only affects constants in the error estimates.

Most of the estimates are done in the region where the confining term is zero.
Due to this, we do not show it in equations unless it is actually used.
\end{remark}

In the $\gamma\beta > 1$ case we get a result that parallels \cite[Theorem~2.6]{berglund_sharp_2013}.
We define the initial and stopping sets in terms of Besov--Hölder norm.

\begin{theorem}[Expected transition time, $\gamma\beta > 1$]\label{thm:main transition time hard}
We consider \eqref{eq:DSG truncated}--\eqref{eq:confining potential} with $\gamma\beta > 1$,
and define for sufficiently small $\kappa, \delta > 0$ the sets
\begin{gather*}
A \coloneqq \left\{ \norm{u + 2\pi/\beta}_{\Holder^{1/2-\kappa}} < \delta \right\}
    \cup \left\{ \norm{u - 2\pi/\beta}_{\Holder^{1/2-\kappa}} < \delta \right\},\\
B \coloneqq \left\{ u \in L^2(\Tor) \colon
    \norm{u}_{\Holder^{1/2-\kappa}} < \delta \right\}.
\end{gather*}
Then there exist $\epsilon_0$, $N_0$, and a probability measure
$\mu_{N,\epsilon}$ on $\partial B$ such that the hitting time $\tau_A$ satisfies
\begin{equation}\label{eq:result hard prefactor}
\E_{\mu_{N,\epsilon}}(\tau_A) =
\frac{1}{2 \norm{\partial_x u_\ast}_{L^2(\Tor)}} \sqrt{ \frac{2\pi}{\mu}
    \frac{\prod_k \lambda_k}{\prod_{\abs n \leq N} n^2 + \gamma\beta}}
    e^{4\pi\gamma/(\epsilon\beta)} (1 + r(\epsilon)),
\end{equation}
whenever $\epsilon < \epsilon_0$ and $N > N_0 \epsilon^{-3}$.
Here $-\mu < 0 < \lambda_1 \leq \ldots \leq \lambda_{2N-1}$ are the eigenvalues of $\Hess F[u_\ast]$
at a transition state $u_\ast$, and
the error satisfies
\[
-\epsilon^{1/2} \log(1/\epsilon)^3 \lesssim r(\epsilon)
    \lesssim \epsilon^{1/4} \log(1/\epsilon)^3,
\]
independently of $N$.
\end{theorem}

In the simpler $\gamma\beta < 1$ case we can choose the initial and stopping sets to shrink as $\epsilon \to 0$.
There is some flexibility in the choice of sets;
we could also add a $\Holder^{1/2-\kappa}$ bound on the oscillatory part $u_\perp$ in the stopping set $A_\epsilon$.

\begin{theorem}[Expected transition time, $\gamma\beta < 1$]\label{thm:main transition time easy}
Consider again the flow of~\eqref{eq:DSG truncated}--\eqref{eq:confining potential} with $\gamma\beta < 1$.
Given any $\epsilon > 0$ and fixed $\kappa > 0$, we define the sets
\begin{gather*}
A_\epsilon \coloneqq \left\{ u \in L^2(\Tor) \colon
    |\hat u(0)| > 2\pi/\beta - \sqrt\epsilon \log(1/\epsilon) \right\},\\
B_\epsilon \coloneqq \left\{ u \in L^2(\Tor) \colon
    |\hat u(0)| < \sqrt\epsilon \log(1/\epsilon),\,
    \norm{u_\perp}_{\Holder^{1/2-\kappa}} < \sqrt{c_0\, \epsilon \log(1/\epsilon)} \right\},
\end{gather*}
where $c_0$ is a sufficiently large constant.
Then there exists a probability measure
$\mu_{N,\epsilon}$ on $\partial B_\epsilon$ such that
\begin{equation}\label{eq:result easy prefactor}
\E_{\mu_{N,\epsilon}}(\tau_{A_\epsilon}) =
\frac{1}{2\gamma\beta} \sqrt{\prod_{0 < \abs n \leq N} \frac{n^2 - \gamma\beta}{n^2 + \gamma\beta}}
    e^{4\pi\gamma/(\epsilon\beta)} (1 + r(\epsilon)),
\end{equation}
where $-\epsilon \lesssim r(\epsilon) \lesssim \epsilon (\log(1/\epsilon))^4$.
These bounds are again uniform in $N$.
\end{theorem}

The square root expressions in~\eqref{eq:result hard prefactor} and~\eqref{eq:result easy prefactor}
are related to functional determinants.
In the $\gamma\beta < 1$ case there is a well-known identity
\begin{equation}
\frac{1}{2\gamma\beta} \sqrt{\prod_{0 < \abs n \leq \infty} \frac{n^2 - \gamma\beta}{n^2 + \gamma\beta}}
= \frac{\sin(\pi \sqrt{\gamma\beta})}{2\gamma\beta \sinh(\pi \sqrt{\gamma\beta})}.
\end{equation}
In the $\gamma\beta > 1$ case we cannot compute the prefactor explicitly.
We can still show that the square root expression stays positive as $N \to \infty$.
Up to a factor involving $\mu$, it is computed in~\eqref{eq:determinant hard}.

\begin{remark}
Convergence of the random variables $\tau_A^{(N)}$ in the SPDE limit $N \to \infty$
is an interesting technical question.
In \cite{berglund_sharp_2013,tsatsoulis_exponential_2020}
it is shown with higher moment estimates and uniform integrability.
This argument should be straightforward to transfer to our setup.
We opt to focus on the capacity computations in this article.
\end{remark}

\begin{remark}
Moreover, it would be natural to state \Cref{thm:main transition time easy,thm:main transition time hard}
for fixed initial data $u_0 = 0$.
This is done in the finite-dimensional setup with the Harnack inequality,
which however depends on the dimension.
In \cite{berglund_sharp_2013,tsatsoulis_exponential_2020} this issue is circumvented
with a post-processing method:
the solution is likely to lose memory of the initial condition
during the long time it spends within a well
(see also \cite{martinelli_small_1988,martinelli_small_1989}).

The argument should again be adaptable to our case,
but goes beyond the scope of this article.
In the $\gamma\beta < 1$ case our measures $\mu_{N,\epsilon}$
do converge weakly to a Dirac measure as $\epsilon \to 0$.
\end{remark}

The study of transition times is motivated by approximation of the system
with a simpler Markov chain.
The exponential loss of memory suggests that times between successive jumps
are asymptotically exponentially distributed
(see also \cite[Section~8.4.4]{bovier_metastability_2015}).
Description of metastable systems as Markov chains is classical~\cite{freidlin_random_1984};
see also~\cite{digesu_sharp_2019} and references therein.

It would then be natural to expect the following asymptotic behaviour of the sine-Gordon system:

\begin{conjecture}
Let us rescale the time in~\eqref{eq:DSG} by the mean transition
time~\eqref{eq:result hard prefactor} or~\eqref{eq:result easy prefactor},
depending on $\gamma\beta$.
Then as $\epsilon \to 0$, the scaled process
approaches a symmetric simple random walk on $\{ k2\pi/\beta \colon k \in \Z \}$,
where the times between jumps are exponentially distributed with mean~$1$.
\end{conjecture}

Quantitative bounds on convergence to a jump process are harder to find,
and there does not appear to be a result directly applicable to our setup
(but see~\cite{sugiura_metastable_1995} for a finite-dimensional result).
Because of this, we leave the meaning of ``approaches'' open in the conjecture above.

It would be very interesting to rigorously state and prove
convergence of metastable SPDE dynamics to a discrete jump process.
The dynamical sine-Gordon model gives an interesting and nontrivial toy example for this question.

\bigskip\noindent
Let us finish this introduction with two further variations on the theme.

We could also consider the \emph{massive} sine-Gordon equation
\begin{equation}
\partial_t u  + (m^2 - \Laplace) u = -\gamma \sin(\beta u) + \sqrt{2\epsilon} \xi
\end{equation}
for some $m^2 > 0$.
In this model each potential well (up to symmetry) is at different height.
This model is interesting in that it has a natural stationary measure
(see e.g.\ \cite{gubinelli_simple_2024} for discussion).
As the transitions are highly biased towards $0$,
the transition rates between arbitrary wells are dominated by lowest-energy wells.
This makes the computation of sharp rates a challenge
\cite[Example~3.5]{berglund_kramers_2013}.

Hyperbolic variants of the equation (e.g.~\cite{oh_twodimensional_2021}) are
related to canonical stochastic quantization.
A double-well wave equation with initial data from the Gibbs measure
has been studied by Newhall, Tal, and Vanden-Eijnden
\cite{tal_transition_2006,newhall_metastability_2017},
and Barashkov and the present author \cite{barashkov_eyringkramers_2024}.
These articles are based on a weaker definition of transitions
(see the discussion in \cite{barashkov_eyringkramers_2024}),
namely the hitting frequency of a hypersurface that separates the potential minima.

Under hyperbolic dynamics the transition time is expected to differ from its parabolic counterpart
only by a factor of $\sqrt \mu$,
where $-\mu$ is the negative eigenvalue at the saddle; see e.g.~\cite{lee_eyringkramers_2025}.

\bigskip\noindent\textbf{Acknowledgements.}
PL would like to thank Nikolay Barashkov, Giacomo Di~Gesù, and Antti Kupiainen
for numerous useful discussions and comments,
and Vesa Julin for discussions and hospitality during a visit to University of Jyväskylä.
PL was financially supported by the Finnish Centre of Excellence in Randomness and Structures (FiRST).

\section{Stochastic and deterministic equation}

\subsection{Spectral Galerkin approximation}

Let us first introduce some notation.
Since the functions are real-valued,
it is natural to equip $L^2$ with the sine-cosine basis
\begin{equation}
e_0(x) = 1,\;
e_{-n}(x) = \frac{1}{\sqrt{\pi}} \sin(nx),\;
e_n(x) = \frac{1}{\sqrt{\pi}} \cos(nx)
\text{ for } n \in \Z_+.
\end{equation}
Here $e_0$ is unnormalized so that $u(x) \equiv \pm 2\pi/\beta$ corresponds to $\hat u(0) = \pm 2\pi/\beta$.
The Parseval identity then takes the form
\begin{equation}
\norm{u}_{L^2}^2
= \int_\Tor \bigg( \sum_{n \in \Z} \hat u(n) e_n(x) \bigg)^{\! 2} \dx
= 2\pi \hat u(0)^2 + \sum_{n \neq 0} \hat u(n)^2.
\end{equation}
We denote by $H^s$ and $\Holder^s$ the Besov spaces
$B^s_{2,2}$ and $B^s_{\infty,\infty}$ respectively;
see e.g.~\cite[Section~2.2]{barashkov_eyringkramers_2024} for their definitions.

We will use the following notation for decomposition of functions into mean and oscillatory parts:
\begin{equation}
u = \bar u + u_\perp,
\text{ where } \bar u(x) = \hat u(0)
\text{ and } u_\perp(x) = \sum_{n \neq 0} \hat u(n) e_n(x).
\end{equation}
The operator $\projN$ truncates the Fourier series to $\abs n \leq N$,
and $\projNperp = \operatorname{Id} - \projN$.
Additionally, $\projperp u = u_\perp$ in the above notation.

In \Cref{sec:transition times} we will implicitly truncate all Fourier series to $\abs n \leq N$.
We then identify a function $u \in L^2(\Tor)$
with its Fourier coefficients $\hat u \in \R^{2N+1}$.

\bigskip\noindent
Let us then consider the wellposedness and approximation of the sine-Gordon equation.
This theory is already classical, so we refer to the following formulation of it:

\begin{theorem}[Wellposedness]\label{thm:wellposedness}
Fix $T > 0$.
For any initial condition $u_0 \in \Holder(\Tor)$,
equation~\eqref{eq:DSG} has a unique mild solution $u \in \Holder([0,T] \times \Tor)$.
Moreover, there is an almost surely finite random variable $Z$ such that
$u$ can be approximated with solutions $u^{(N)}$ to \eqref{eq:DSG truncated} with
\[
\sup_{0 \leq t \leq T} \norm{u(t, \omega) - u^{(N)}(t, \omega)}_{L^\infty} \leq Z(\omega) N^{-c},
\]
where $\omega$ is the realization of the noise.
\end{theorem}
\begin{proof}
\cite[Theorem~3.1 and Section~4.2]{blomker_galerkin_2013}.
\end{proof}

The estimate above is not very precise.
It can be improved to bounded expectation with a \mbox{Grönwall} argument.
We make a simplifying assumption on the initial data,
which matches the usage of this lemma in \Cref{thm:ldp equilibrium local}.

\begin{lemma}[Galerkin approximation]\label{thm:galerkin approximation}
Consider the solutions $u$ to \eqref{eq:DSG} and $u^{(N)}$ to \eqref{eq:DSG truncated},
where both have Fourier-truncated initial data $u_0 = \projN u_0$ and the same realization of the noise.
Then for any $\kappa > 0$ and $0 \leq \alpha < 1/2 - \kappa$ we have
\[
\E \norm{u(t) - u^{(N)}(t)}_{\Holder^\alpha(\Tor)}
\lesssim N^{\alpha - 1/2 + \kappa} C(t),
\]
uniformly in $N$, where $C(t)$ is continuous and increasing in $t$.
\end{lemma}
\begin{proof}
By the definition of mild solution,
\begin{equation}
u(t) - u^{(N)}(t)
= -\gamma \int_0^t \!e^{(t-s)\Laplace} \!\left[
    \sin(u(s)) - \projN \sin(u^{(N)}(s))
\right] \diff s
    + \sqrt{2\epsilon} \projNperp(\mathcal O_t),
\end{equation}
where $e^{t \Laplace}$ is the heat semigroup
and $\mathcal O_t = (e^{\cdot \Laplace} \ast \xi)(t)$ is the stochastic convolution of the white noise.
The linear $e^{t\Laplace} u_0$ terms are cancelled out.

We can then use Besov embeddings and Bernstein's theorem to estimate
\begin{align}
\norm{u(t) - u^{(N)}(t)}_{\Holder^\alpha}
&\lesssim \int_0^t \norm{\projN e^{(t-s)\Laplace} [\sin(u(s)) - \sin(u^{(N)}(s))]}_{H^{\alpha+1/2+\kappa}} \!\!\diff s \notag\\
    &\qquad + N^{\alpha - 1/2 + \kappa} \int_0^t \norm{\projNperp e^{(t-s)\Laplace} \sin(u(s))}_{H^1} \diff s\\
    &\qquad+ \sqrt{2\epsilon} N^{\alpha - 1/2 + \kappa} \norm{\mathcal O_t}_{\Holder^{1/2-\kappa}}. \notag
\end{align}
The heat semigroup satisfies
$\|{e^{t\Laplace} \phi}\|_{H^{r}} \leq (1 + C_r t^{-r/2}) \norm{\phi}_{L^2}$
for any $r \geq 0$ (see e.g.\ \cite[Lemma~3.3.3]{berglund_introduction_2022}).
Using this and Lipschitz continuity of sine we get
\begin{align}
\norm{u(t) - u^{(N)}(t)}_{\Holder^\alpha}
&\lesssim \int_0^t (1 + (t-s)^{-\alpha/2-1/4-\kappa/2}) \norm{u(s) - u^{(N)}(s)}_{L^\infty} \diff s \notag\\
    &\qquad + N^{\alpha - 1/2 + \kappa} \int_0^t (1 + (t-s)^{-1/2}) \diff s\\
    &\qquad+ \sqrt{2\epsilon} N^{\alpha - 1/2 + \kappa} \norm{\mathcal O_t}_{\Holder^{1/2-\kappa}}. \notag
\end{align}
Then Grönwall's inequality implies
\begin{equation}
\begin{split}
\norm{u(t) - u^{(N)}(t)}_{\Holder^\alpha}
&\lesssim N^{\alpha - 1/2 + \kappa} [C(t) + \sqrt{2\epsilon} \norm{\mathcal O_t}_{\Holder^{1/2-\kappa}}]\\
    &\qquad \exp\!\left( \int_0^t (1 + (t-s)^{-\alpha/2-1/4-\kappa/2}) \diff s \right)\!.
\end{split}
\end{equation}
The expectation of the $\norm{\mathcal O_t}$ term is bounded by
parabolic Hölder regularity of the stochastic convolution
\cite[Corollary~3.3.11]{berglund_introduction_2022},
so we get the claim.
\end{proof}

\subsection{Gates and stationary solutions}

The potential-theoretic approach requires us to consider the topology of the energy landscape.
The lowest-energy path between two potential minima is described by the mountain pass theorem.
See e.g.\ \cite[Section~8.5]{evans_partial_2002} for a (simplified) proof.

\begin{theorem}[Mountain pass theorem]\label{thm:mountain pass}
Let $H$ be a Hilbert space, $u_0 \in H$, and $I \in C^1(H;\, \R)$ a functional such that
\begin{enumerate}
\item $I$ satisfies the Palais--Smale compactness condition:
    if $I[u_k]$ is bounded and $I'[u_k] \to 0$, then the sequence $u_k$ has a convergent subsequence;
\item there exist $r, \delta > 0$ such that $I[u] \geq I[u_0] + \delta$ whenever $\norm{u - u_0} = r$; and
\item there is $v \in H$ such that $\norm{v - u_0} > r$ and $I[v] \leq I[u_0]$.
\end{enumerate}
Let $\Gamma$ be the set of $C([0,1];\, H)$ paths connecting $u_0$ to $v$.
Then
\[
\inf_{\gamma \in \Gamma} \max_{t \in [0,1]} I[\gamma(t)]
\]
is attained at a critical point of $I$.
\end{theorem}

This connects well with the following definition,
as used by \cite{berglund_eyringkramers_2010,bovier_metastability_2015,avelin_geometric_2023}.
It is important to notice that the mountain pass might not be unique;
correspondingly the gates are sets of points.
In our case the essential gate will be a single point if $\gamma\beta < 1$,
and a closed curve if $\gamma\beta > 1$.

\begin{definition}[Essential gate]
A path $\gamma \in \Gamma$ is \emph{optimal}
and $c \in \R$ the \emph{communication height between $u_0$ and $v$} if
\[
\max_{t \in [0,1]} I[\gamma(t)]
= \inf_{\eta \in \Gamma} \max_{t \in [0,1]} I[\eta(t)]
\eqqcolon c.
\]
Then $G \subset H$ is a \emph{gate between $u_0$ and $v$}
if it is a minimal subset of $I^{-1}\{c\}$ such that all optimal paths intersect it.
Gates might not be unique; the union of all gates is called an \emph{essential gate}.
These definitions extend naturally to $u_0, v \subset H$ instead of points.
\end{definition}

In order to use the mountain pass theorem with the sine-Gordon model,
we need the potential to be confining also in the direction of the zero mode,
as defined in~\eqref{eq:confining potential}.
Otherwise $u_k \equiv k\pi/\beta$ would be a counterexample.

\begin{lemma}[Existence of essential gates]\label{thm:sg mountain pass}
\Cref{thm:mountain pass} holds for the
localized sine-Gordon model~\eqref{eq:DSG truncated}--\eqref{eq:confining potential},
and hence the essential gate between adjacent minima is a subset of the critical points of $F$.
\end{lemma}
\begin{proof}
The Palais--Smale condition follows from the potential having quadratic growth in every direction:
The assumption $|F[u_k]| \leq K$ implies
\begin{equation}
\frac 1 2 \norm{\nabla u_k}_{L^2}^2 - \frac{2\pi\gamma}{\beta}
   + \frac 1 2 \abs{\hat u_k(0)}^2 - (4\pi/\beta)^2 \leq K,
\end{equation}
and hence the $H^1$ norm of $u_k$ is uniformly bounded.
As $H^1(\Tor)$ is compactly embedded into $L^2(\Tor)$,
there is an $L^2$-convergent subsequence.

By symmetry, we can choose $u_0 = 0$
and $v$ to be the constant $2\pi/\beta$ function.
Then $F[u_0] = F[v] = -2\pi\gamma/\beta$.

It remains to verify the third condition.
We pick $r = \pi/\beta$.
Let $u \in L^2$ satisfy $\norm{\hat u}_{\ell^2} = r$.
We split into two cases based on $\norm{u_\perp}$.
If $\norm{u_\perp}_{\ell^2} \geq \pi/(4\beta)$, then it automatically follows that
\begin{equation}
F[u] \geq \frac 1 2 \left(\frac{\pi}{4\beta}\right)^{\!2} - \frac{2\pi\gamma}{\beta}.
\end{equation}
Let us then consider the case $\norm{u_\perp}_{\ell^2} \leq \pi/(4\beta)$.
We now need a non-trivial estimate on the cosine term.
The assumption implies
\begin{equation}
\sqrt{\frac{15}{16}} \frac{\pi}{\beta} \leq \abs{\hat u(0)} \leq \frac{\pi}{\beta}.
\end{equation}
This then implies the pointwise estimate $\cos(\beta u(x)) < 0$
when $\abs{u_\perp(x)} < \pi/(5\beta)$.
We can hence use Markov's and Hölder's inequalities to estimate
\begin{equation}
F[u]
\geq -\frac{\gamma}{\beta}
    \big|{\{ x \in \Tor \colon \abs{u_\perp(x)} \geq \pi/(5\beta) \}}\big|
\geq -\frac\gamma\beta \cdot \frac{5\beta \cdot \sqrt{2\pi} \norm{u_\perp}_{L^2}}{\pi}.
\end{equation}
By assumption then
\begin{equation}
F[u]
\geq -\frac{5}{4\sqrt{2\pi}} \cdot \frac{2\pi\gamma}{\beta}
> -\frac{1}{2} \cdot \frac{2\pi\gamma}{\beta},
\end{equation}
which is again strictly greater than the potential at a minimum.
\end{proof}

Let us then derive a concrete expression for the critical points.
If we set $\partial_t u = 0$ and $\epsilon = 0$ in \eqref{eq:DSG}, we get
\begin{equation}
u''(x) = \gamma \sin(\beta u(x)).
\end{equation}
This is a Hamiltonian system in the $x$ variable,
and its phase portrait is shown in \Cref{fig:phase portrait}.
All orbits are stationary solutions,
but only those whose lengths divide $2\pi$ satisfy the periodic boundary condition.

\begin{figure}[tb]
\centering
\includegraphics[width=\textwidth]{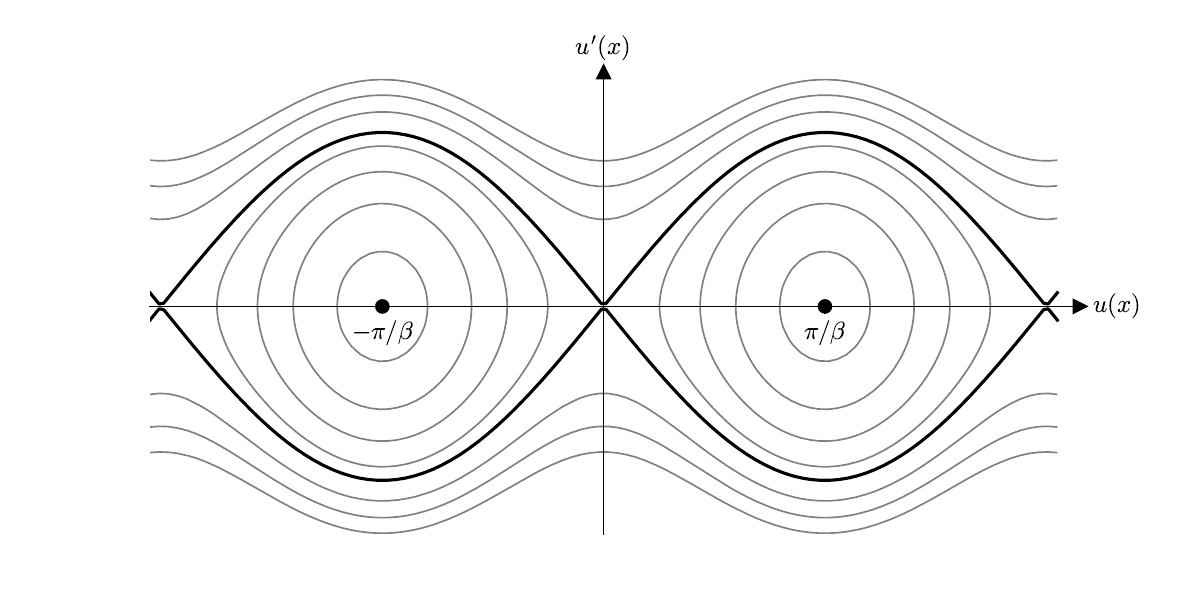}
\caption{Phase portrait of the stationary system.
The highlighted contour corresponds to the heteroclinic orbit with infinite period,
and the filled circles to period-$0$ orbits.}
\label{fig:phase portrait}
\end{figure}

With some manipulation the equation can be converted into the pendulum equation
\begin{equation}
f''(x) = -\gamma\beta \sin(f(x)), \quad\text{with } f(x) = \beta u(x) - \pi,
\end{equation}
which yields the explicit solution
\begin{equation}\label{eq:stationary solution}
u(x) = \frac{1}{\beta} \left[
    \pi + 2 \arcsin(\sqrt m \cd(\sqrt{\gamma\beta} x, m)) \right],
\quad m = \cos(\beta u(0)/2)^2.
\end{equation}
Here $\cd$ is the Jacobi elliptic function.
In our parametrization the period of $\cd$ is $4K(m)$,
where $K$ is Jacobi elliptic integral function.

Since $K$ is increasing, we find that the period $P(u_0)$ of an orbit passing through $(u_0, 0)$
is increasing in $u_0 \in [\pi/\beta, 2\pi/\beta]$ and satisfies the limits
\begin{equation}
P(\pi/\beta) = \frac{2\pi}{\sqrt{\gamma\beta}},
\quad P(2\pi/\beta) = \infty.
\end{equation}
New stationary solutions with period dividing $2\pi$ hence appear
when $\sqrt{\gamma\beta} \in \N$.
See \Cref{fig:solutions} for examples.
This behaviour is similar to that of the Allen--Cahn equation~\cite{maier_effects_2003}.

\begin{figure}[tb]
\centering
\includegraphics[width=0.75\textwidth]{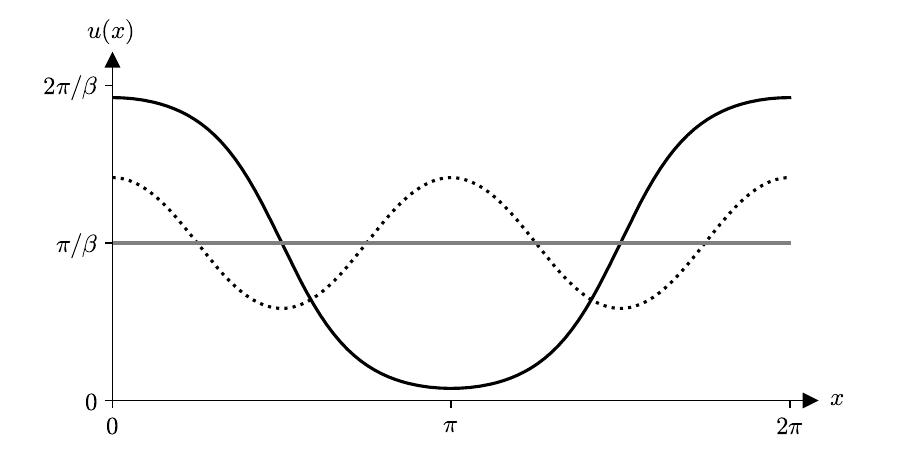}
\caption{The three (representatives of) unstable stationary solutions when $\gamma = 1$ and $\beta = 5$.
The solid black line is the energetically optimal transition state.}
\label{fig:solutions}
\end{figure}

\subsection{Transition states}\label{sec:stationary}

By the theory so far, we know that all points in the essential gate are critical points of the potential.
We still need to exclude the critical points that are not part of an essential gate.
The equivalent characterization
used in \cite{berglund_eyringkramers_2010,berglund_sharp_2013} is more convenient for this.
See \cite[Section~2]{berglund_eyringkramers_2010} for comparison of the two approaches.

\begin{definition}[Transition state]
A critical point $z$ of potential $F \in C^2(\R^d)$ is a \emph{transition state}
if $\Hess F[z]$ has \emph{exactly} one negative eigenvalue.
\end{definition}

\begin{lemma}[Eigenvalues at essential gates]
Assume that the Hilbert space $H$ is finite-dimensional and $I \in C^2(H;\, \R)$.
Let $u$ belong to an essential gate.
Then $\Hess I[u]$ has \emph{at most} one strictly negative eigenvalue.
\end{lemma}
\begin{proof}
\cite[Proposition~2.6]{berglund_eyringkramers_2010}.
\end{proof}

The potential energies of the stationary solutions
form the bifurcation diagram of \Cref{fig:energies}.
The kink-antikink pair that appears at $\gamma\beta = 1$ is energetically optimal
for all $\gamma\beta > 1$.
We prove this rigorously for the truncated system in \Cref{thm:transition states}.

\begin{figure}[tb]
\centering
\includegraphics[width=0.75\textwidth]{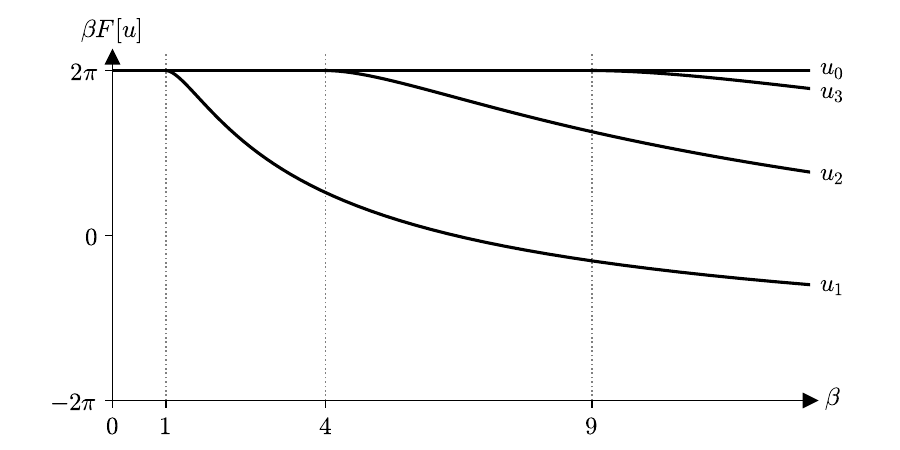}
\caption{The potential energies of the stationary solutions follow a bifurcation diagram.
The $2\pi$-periodic stationary solution that appears at $\gamma\beta = 1$ has the lowest energy.
Here $\gamma = 1$.
Note that the energy on $y$-axis is scaled by $\beta$.}
\label{fig:energies}
\end{figure}

As the potential-theoretic development in \Cref{sec:transition times}
requires finitely many dimensions,
we will now pass to the truncated equation~\eqref{eq:DSG truncated}.
This does not affect \Cref{thm:sg mountain pass}.
The following lemma shows that the approximation is valid in our context.

\begin{lemma}[Convergence of stationary solutions]\label{thm:critical galerkin}
Let $N$ be sufficiently large.
Up to translates (in the case $\gamma\beta > 1$),
the full equation~\eqref{eq:DSG} and truncated equation~\eqref{eq:DSG truncated}
have the same number of stationary solutions.

Moreover, for any $p > 0$ the following holds:
If $u_\ast$ is a stationary solution of~\eqref{eq:DSG},
then~\eqref{eq:DSG truncated} has a stationary solution $u_N$ such that
$\norm{u_\ast - u_N}_{L^2} \lesssim N^{-p}$ for all $N \geq N_0$.
The constants depend on $\gamma$, $\beta$, and $p$.
\end{lemma}
\begin{proof}
The qualitative statement is \cite[Proposition~4.9]{berglund_sharp_2013},
with passing to normal space in the case of degenerate saddles.
We therefore only derive the quantitative estimate
on critical points of the potential~\eqref{eq:confining potential}
in either $\ell^2(\Z)$ or $\R^{2N+1}$.

In the following, we indicate partial derivatives of $F$ by subscripts.
Let $x_0 = \projN u_\ast$ identified as an element of $\R^{2N+1}$,
and $y_0 = \projNperp u_\ast$.
Starting from $F_x(x_0, y_0) = 0 \in \R^{2N+1}$, our goal is to find low modes $x$
such that $F_x(x, 0) = 0$.

We use \cite[XVIII.\S1, Theorem~6]{kantorovic_functional_1982} that gives sufficient conditions
for convergence of Newton iteration.
These conditions are:
\begin{enumerate}
\item \emph{$F_{xx}(x_0, 0)^{-1}$ exists and is bounded.}
    By assumption
    \begin{equation}
    \Hess F(x_0, y_0) = -\Laplace + \gamma\beta \cos(\beta u_\ast)
    \end{equation}
    projected to the non-degenerate modes is invertible.
    Our operator of interest is
    \begin{equation}
        F_{xx}(x_0, 0) = -\projN\Laplace + \gamma\beta \projN \cos(\beta \projN u_\ast) \projN.
    \end{equation}
    Since $u_\ast$ is smooth, we have $\norm{(1-\projN) u_\ast}_{L^\infty} \leq C_p N^{-p-3}$
    for any $p > 0$.
    Lipschitz continuity of the cosine then implies
    \begin{equation}
        \norm{F_{xx}(x_0, 0) - \projN \Hess F(x_0, y_0) \projN}_{L^2 \to L^2} \lesssim N^{-p-3},
    \end{equation}
    and by choosing $N$ large enough, the eigenvalues of $F_{xx}$ are bounded away from $0$.
    The bound can be chosen uniformly in $N \geq N_0$.

\item \emph{$\norm{F_{xx}(x_0, 0)^{-1} F_x(x_0, 0)} \leq \eta$.}
    By Fréchet derivatives,
    \begin{equation}
        [F_x(x_0, 0)] v = \int_\Tor
            [-\Laplace u_\ast(x) + \gamma \sin(\beta \projN u_\ast(x))] \projN v(x) \dx
    \end{equation}
    for any $v \in L^2(\Tor)$.
    Again by Lipschitz continuity, the operator norm satisfies
    \begin{equation}
        \norm{\projN [F_x(x_0, 0) - \nabla F(x_0, y_0)]}
        \lesssim N^{-p-3}.
    \end{equation}
    By the uniform bound on $F_{xx}(x_0, 0)^{-1}$, we can then pick $\eta = C_p N^{-p-3}$.

\item \emph{$\norm{F_{xx}(x_0, 0)^{-1} F_{xx}(x, 0)} \leq K$ for all $\norm{x} \leq r$.}
    Finally for generic $x$ and $v \in L^2$ we have the trivial bound
    \begin{equation}
        \norm{F_{xx}(x, 0) v}_{L^2}
        \leq \norm{-\projN\Laplace v}_{L^2} + \gamma\beta \norm{\projN v}_{L^2}
        \lesssim N^2 \norm{v}_{L^2}.
    \end{equation}
    Hence we can pick $K = C_p N^2$ and e.g.\ $r = 1$.
\end{enumerate}
Now since $K\eta \leq 1/2$ for $N$ large enough,
\cite[XVIII.\S1, Theorem~6]{kantorovic_functional_1982}
implies that Newton iteration started from $x_0$
converges to a solution $x_N$ such that
\begin{equation}
\norm{x_0 - x_N}_{L^2} \leq \frac{1 - \sqrt{1 - 2K\eta}}{K},
\end{equation}
assuming that the right-hand side is less than $r$.
In our case this is true, so
\begin{equation}
\norm{x_0 - x_N}_{L^2} \leq \frac{2K\eta}{K} \lesssim N^{-p}.
\qedhere
\end{equation}
\end{proof}

\begin{proposition}[Sine-Gordon transition states]\label{thm:transition states}
The transition states, and hence essential gates, of the localized and truncated sine-Gordon model are
\begin{itemize}
    \item the constant functions $\{ -\pi/\beta, \pi/\beta \}$ when $\gamma\beta < 1$;
    \item the period-$2\pi$ stationary solutions when $\gamma\beta > 1$.
\end{itemize}
The Hessian matrix of $F$ at the transition state has no vanishing eigenvalues when $\gamma\beta < 1$,
and exactly one vanishing eigenvalue when $\gamma\beta > 1$.
\end{proposition}
\begin{proof}
In the case $\gamma\beta < 1$ there is only one unstable stationary solution,
and the eigenvalues of $\Hess F[\pi/\beta]$ are $k^2 - \gamma\beta$ for all $\abs k \leq N$.
At $\gamma\beta = 1$ we get two vanishing eigenvalues, and a bifurcation occurs.

The kink-antikink pair $u_1$ that appears from \eqref{eq:stationary solution}
must then be the transition state.
Since both $u_1(x)$ and $u_1(t+x)$ solve the equation,
$\Hess F[u_1]$ has a zero eigenvalue corresponding to eigenvector $u_1'(x)$.

Further bifurcations at $\pi/\beta$ cannot produce transition states
as there are several negative eigenvalues.
On the other hand the number of negative eigenvalues of $\Hess F[u_1]$ cannot change:
By continuity of the eigenvalues as a function of the matrix
(see \cite[Section~II.\S5]{kato_perturbation_1995}),
a sign change would be associated with additional zero eigenvalue.
That would imply two orbits in \Cref{fig:phase portrait} crossing, which is not possible.
\end{proof}

\section{Estimating transition times}\label{sec:transition times}

\subsection{Potential-theoretic approach}\label{sec:potential theory}

As mentioned in the introduction, the exponential asymptotics of transition times
were first studied by Fre\u{\i}dlin and Wentzell \cite{freidlin_random_1984}.
We present a brief summary of their method in \Cref{sec:ldp},
as we use it to bootstrap a better estimate in the $\gamma\beta > 1$ case.

The potential-theoretic approach of Bovier, Eckhoff, Gayrard, and Klein \cite{bovier_metastability_2004}
is based on computing the ratio of two quantities:
the $L^1$ integral of the equilibrium potential and the capacity.
Roughly put, the first quantity sees the geometry of the potential wells,
and the second quantity the geometry near the saddle.

\begin{definition}[Equilibrium potential]
Given two disjoint sets $X, Y \subset \R^d$, we define
\[
h_{X, Y}(z) = \Prob_z(\tau_X < \tau_Y),
\]
where $\tau_X$ and $\tau_Y$ are the first hitting times of sets $X$ and $Y$ respectively
under the diffusion started at $z \in \R^d$.
\end{definition}

\begin{definition}[Capacity]\label{def:capacity}
The capacity of two disjoint sets $X, Y \subset \R^d$ is defined in three equivalent ways
(see \cite[Lemma~2.4]{avelin_geometric_2023}):
\begin{align*}
\capa(X, Y)
&= \inf_{g} \epsilon \int_{\R^d} \abs{\nabla g(z)}^2 \exp(-F(z)/\epsilon) \dz\\
&= \epsilon \int_{\R^d} \abs{\nabla h_{X,Y}(z)}^2 \exp(-F(z)/\epsilon) \dz\\
&= \int_{\partial X} \exp(-F(z)/\epsilon) e_{X,Y}(\diff z). 
\end{align*}
On the first line the infimum is taken over $g \in H^1(\R^d)$ such that
$g|_X \geq 1$ and $g|_Y = 0$.
The measure $e_{X,Y}$ on $\partial X$ is the \emph{equilibrium measure}.
\end{definition}

Hitting times of reversible Markov diffusions can be expressed as elliptic PDE problems,
which in turn can be analyzed with potential theory.
A good summary of the derivation is given in \cite[Section~3.1]{berglund_sharp_2013}.
The hitting time $\tau_A$ of a set $A \subset \R^d$ satisfies the relation
\begin{equation}
\int_{A^\complement} h_{B,A}(z) \exp(-F(z)/\epsilon) \dz
= -\int_{\partial B} [\E^z \tau_A]  \exp(-F(z)/\epsilon) e_{B,A}(\dz).
\end{equation}
In the method of \cite{bovier_metastability_2004},
one then uses the Harnack inequality to estimate the right-hand side by
$[\E^x \tau_A] \capa(A, B)$, when $B$ is an $\epsilon$-ball around $x \in \R^d$.

The Harnack inequality depends on the dimension,
and hence cannot be used for uniform-in-$N$ bounds.
Berglund and Gentz \cite{berglund_sharp_2013} instead consider initial data from the measure
that comes from normalizing the right-hand side.
This leads to the following statement:

\begin{theorem}[Expected transition time]
\begin{equation}
\E_{\mu} \tau_A = \frac{\int_{A^\complement} h_{B,A}(z) \exp(-F(z)/\epsilon) \dz}{\capa(A,B)},
\end{equation}
where $\E_\mu$ is over the probability measure $\mu(\diff z) \propto \exp(-F(z)) e_{B,A}(\diff z)$
supported on $\partial B$.
\end{theorem}

\subsection{Decomposition of the potential}

Let us then present the choice of hitting sets in the $\gamma\beta < 1$ case.
We make a different choice in the $\gamma\beta > 1$ case, presented in \Cref{sec:hard capacity},
but the underlying idea is the same.

Let us fix small $\kappa > 0$.
We take the initial data from the set
\begin{gather}
B_\epsilon \coloneqq \{ u \in L^2(\Tor) \colon |\hat u(0)| < \sqrt\epsilon \log(1/\epsilon),
    \projperp u \in B_\epsilon^\perp \}, \text{ where}\\
B_\epsilon^\perp \coloneqq \bigg\{
    u \in L^2(\Tor) \colon
    \hat u(0) = 0,\; \norm{u}_{\Holder^{1/2-\kappa}} < \sqrt{c_0 \epsilon \log(1/\epsilon)}
\bigg\}.
\end{gather}
That is, $u \in B_\epsilon$ belongs to a vanishing ball in the Besov--Hölder space $\Holder^{1/2-\kappa}(\Tor)$.
From here on we abbreviate the space as $\Holder^{1/2-}$.
We make the oscillatory part vanish slightly faster as $\epsilon \to 0$
in order to control the oscillations with the zero mode.
The constant $c_0 > 0$ will be chosen to be large enough to satisfy
all usages of \Cref{thm:well concentration} below.

The neighbouring potential minima are contained in the strips
\begin{equation}
\begin{gathered}
A_\epsilon^- \coloneqq \{ u \colon \hat u(0) < -2\pi/\beta + \sqrt\epsilon \log(1/\epsilon) \},\\
A_\epsilon^+ \coloneqq \{ u \colon \hat u(0) > 2\pi/\beta - \sqrt\epsilon \log(1/\epsilon) \}.
\end{gathered}
\end{equation}
We are therefore interested in the hitting time of
$A_\epsilon \coloneqq A_\epsilon^- \cup A_\epsilon^+$.

\bigskip\noindent
We compute the relevant quantities as Gaussian integrals.
To this end, we need to Taylor expand the potential at the vicinity of a critical point.

Let $u_\ast$ be a critical point of the potential.
We can use the sum-of-angles formula to decompose $F(u_\ast + v)$ as
\begin{equation}\label{eq:potential decomposition initial}
\begin{split}
&\int_\Tor \frac{\abs{\nabla u_\ast}^2 + 2 \nabla u_\ast \nabla v + \abs{\nabla v}^2}{2}\\
    &\hspace{3em}- \frac{\gamma}{\beta} \cos(\beta u_\ast) \cos(\beta v)
    + \frac{\gamma}{\beta} \sin(\beta u_\ast) \sin(\beta v) \dx.
\end{split}
\end{equation}
The Taylor expansion of $\cos(\beta v)$ yields the term
\begin{equation}
\int_\Tor \frac{\abs{\nabla v}^2}{2}
    - \frac{\gamma\beta}{2} \cos(\beta u_\ast) v(x)^2 \dx
= \frac 1 2 ([-\Laplace + \gamma\beta \cos(\beta u_\ast)] v, v),
\end{equation}
which will give a Gaussian covariance once we remove the negative eigenvalue.
By the stationary equation $\Laplace u_\ast = \gamma \sin(\beta u_\ast)$ we have
\begin{equation}
\int_\Tor \nabla u_\ast \nabla v + \gamma \sin(\beta u_\ast) v \dx = 0.
\end{equation}
Thus
\begin{equation}\label{eq:potential decomposition generic}
F(u_\ast + v)
= F(u_\ast) + \frac 1 2 (\cova v, v)
    + \int_\Tor R(v(x)) \dx,
\end{equation}
where $\cova = -\Laplace + \gamma\beta \cos(\beta u_\ast)$
and $R(z) = \bigO(z^3) + \bigO(z^4)$.
Since the operator has a basis of orthogonal eigenfunctions,
we can further split the inner product to the unstable and stable directions if necessary.

We can improve the error term at potential minima
or when the transition state is a constant function.
This is used in the following two subsections.

\begin{lemma}[Decomposition at constant function]\label{thm:easy potential decomposition}
When $u_\ast \equiv k\pi/\beta$ for some $k \in \Z$, we have
\[
F(u_\ast + v)
= F(u_\ast) + \frac 1 2 (\cova v, v)
    + \operatorname{sign}(\cos(\beta u_\ast)) \int_\Tor R(v(x)) \dx,
\]
where $0 \leq R(z) \lesssim z^4$.
\end{lemma}
\begin{proof}
In this case
the $\nabla u_\ast$ and sine terms in \eqref{eq:potential decomposition initial} vanish,
and $\cos(\beta u_\ast) \equiv \pm 1$.
The Taylor expansion of $\cos(\beta v)$ is more precisely
\begin{equation}
\frac\gamma\beta \cos(\beta v)
= \frac\gamma\beta - \frac{\gamma\beta}{2} v^2 + R(v),
\text{ where } 0 \leq R(v) \leq \frac{\gamma\beta^3 v^4}{24}.
\qedhere
\end{equation}
\end{proof}

\subsection{Potential well}\label{sec:well}

We first show Gaussian concentration of measure on the set $B_\epsilon$.
The constant $c_0$ plays the same role as in \cite[Eq.~(6.26)]{berglund_sharp_2013}.
As the result is used in several lemmas, we choose $c_0$ large enough to satisfy all of them.

This lemma is classical.
Our proof is based on the variational formula introduced by
Boué and Dupuis \cite{boue_variational_1998} and Üstünel \cite{ustunel_variational_2014},
and applied to quantum field theory by Barashkov and Gubinelli \cite{barashkov_variational_2020}.

\begin{lemma}[Gaussian concentration on $B_\epsilon$]\label{thm:well concentration}
We define the $\epsilon^{-1/2}$-dilated set
\[
\Blog \coloneqq \{ u \in L^2(\Tor) \colon \epsilon^{1/2} u \in B_\epsilon \}.
\]
Let $\rho$ be a Gaussian measure of covariance $\projN (-\Laplace + m^2)^{-1}$ for some $m^2 > 0$,
or of covariance $-\!\projperp\!\projN\Laplace^{-1}$.
Then for any $p > 0$ there exist $c_0$ and $C > 0$, dependent on $p$ and the covariance, such that
\[
\int_{\Blog} \diff\rho \geq 1 - C \epsilon^p.
\]
\end{lemma}
\begin{proof}
By construction there is a $\Holder^{1/2-}$-ball of radius $r = \sqrt{c_0 \log(1/\epsilon)}$
contained in $\Blog$ once $\epsilon$ is small enough.
We can then estimate
\begin{equation}
\begin{split}
\int_{B_\epsilon} \diff\rho(u)
&\geq 1 - \int_{B(0,r)^\complement} \hspace{-1em} \diff\rho(u)\\
&\geq 1 - \int_{\R^{2N+1}} \hspace{-1em}
    \exp\!\left( -\lambda\pospart{r^2 - \norm{u}_{\Holder^{1/2-}}^2} \right) \diff\rho(u),
\end{split}
\end{equation}
where $\pospart{\,\cdot\,} \coloneqq \max(0, \cdot\,)$ and $\lambda > 0$ is fixed below.
To bound the last integral from above, we use \cite[Theorem~2]{barashkov_variational_2020}
and find a lower bound for
\begin{equation}
\begin{split}
&\mathrel{\phantom{\geq}} \inf_{v \in \mathbb H_a} \E_\phi \left[ 
    \lambda \pospart{r^2 - \norm{\phi_T + I_T[v]}_{\Holder^{1/2-}}^2} + \frac 1 2 \int_0^T \norm{v_t}_2^2 \dt
\right]\\
&\geq \inf_{v \in \mathbb H_a} \E_\phi \left[ 
    \lambda r^2 - 2\lambda \norm{\phi_T}_{\Holder^{1/2-}}^2
    - 2\lambda \norm{I_T[v]}_{\Holder^{1/2-}}^2 + \frac 1 2 \int_0^T \norm{v_t}_2^2 \dt
\right].
\end{split}
\end{equation}
Here $\phi$ is sampled from $\rho$,
and $\phi_T$ is its regularization as in \cite{barashkov_variational_2020}.
The other notations match those of \cite{barashkov_variational_2020,barashkov_eyringkramers_2024} as well.
It can be shown by hypercontractivity that $\E \norm{\phi_T}_{\Holder^{1/2-}}^2 \lesssim 1$
uniformly in both $T$ and $N$.

For the third term we use the embedding of $H^1$ into $\Holder^{1/2-}$.
By fixing $\lambda$ small enough (depending on $m^2$)
and applying \cite[Lemma~3.7]{barashkov_eyringkramers_2024},
its sum with the $t$-integral is then positive.
Hence for $r^2$ large enough we have
\begin{equation}
\begin{split}
\int_{\R^{2N+1}} \hspace{-1em} \exp\!\left( -\lambda\pospart{r^2 - \norm{v}_{H^{1/2-}}^2} \right) \diff\rho(v)
&\leq \exp(C\lambda - \lambda r^2)\\
&= \exp(C\lambda - \lambda c_0 \log(1/\epsilon)).
\end{split}
\end{equation}
Once we make $\lambda c_0 > p$,
the measure of $B(0,r)^\complement$ is at most $C \epsilon^p$.
\end{proof}

Then we can compute the integral of the equilibrium potential.
The concentration result transfers to the sine-Gordon measure,
so the integral is dominated by the behaviour near the potential minimum.
The dominant error term comes from \Cref{thm:easy potential decomposition}.

\begin{proposition}[Potential well, lower bound]\label{thm:well lower}
We have
\[
\int_{A_\epsilon^\complement} h_{B_\epsilon,A_\epsilon}(u) e^{-F(u)/\epsilon} \du
\geq \sqrt{\frac{\epsilon}{\gamma\beta} \prod_{0 < \abs n \leq N} \frac{2\pi \epsilon}{n^2 + \gamma\beta}}
    e^{2\pi\gamma/(\epsilon\beta)}
    (1 - C \epsilon).
\]
The constant $C$ is independent of $N$.
\end{proposition}
\begin{proof}
Let us restrict the integral to the set $B_\epsilon$ only.
In this set $h_{B,A} \equiv 1$ by definition.
By \Cref{thm:easy potential decomposition} the potential can be expanded as
\begin{equation}
-\frac{F(u)}{\epsilon}
= -\frac{2\pi \gamma\beta \hat u(0)^2}{2\epsilon}
    - \frac{1}{2\epsilon} \sum_{n \neq 0} (n^2 + \gamma\beta) \hat u(n)^2
    + \frac{2\pi\gamma}{\epsilon\beta}
    + \frac 1 \epsilon \int_{\Tor} R(u(x)) \dx.
\end{equation}
Since we are computing a lower bound,
we can ignore the non-negative $R(u)$ term.
After rescaling $u = \sqrt\epsilon v$, the integral is therefore bounded from below by
\begin{equation}\label{eq:easy well evaluated}
e^{2\pi\gamma/(\epsilon\beta)}
\sqrt{\frac{\epsilon}{\gamma\beta} \prod_{n \neq 0} \frac{2\pi \epsilon}{n^2 + \gamma\beta}}
\int_{\Blog} \diff\rho(v),
\end{equation}
where $\rho$ is a Gaussian measure with covariance $(-\Laplace + \gamma\beta)^{-1}$,
restricted to wavenumbers at most $N$,
and $\Blog$ is the $\epsilon^{-1/2}$-dilatation defined in \Cref{thm:well concentration}.
The Gaussian concentration of measure finishes the proof.
\end{proof}

\begin{figure}
\centering
\includegraphics[width=\textwidth]{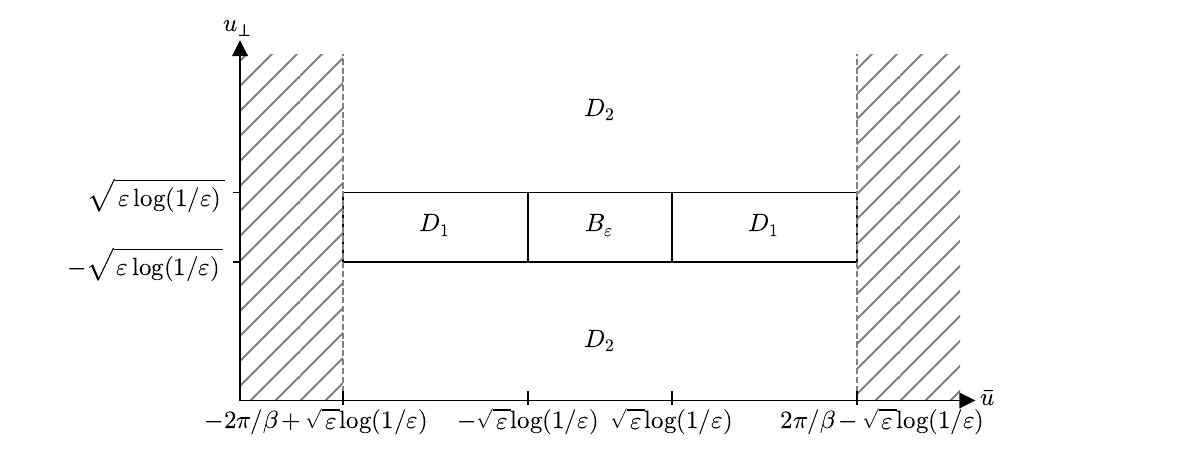}
\caption{A schematic of the phase space decomposition in \Cref{thm:well upper}.
The horizontal axis represents the mean and the vertical axis the oscillatory modes.
In $B_\epsilon$ they vanish at different rates as $\epsilon \to 0$.
The shaded region is $A_\epsilon$.}
\label{fig:well upper schematic}
\end{figure}

\begin{proposition}[Potential well, upper bound]\label{thm:well upper}
Similarly to the above,
\[
\int_{A_\epsilon^\complement} h_{B_\epsilon,A_\epsilon}(u) e^{-F(u)/\epsilon} \du
\leq \sqrt{\frac{\epsilon}{\gamma\beta} \prod_{0 < \abs n \leq N} \frac{2\pi\epsilon}{n^2 + \gamma\beta}}
    e^{2\pi\gamma/(\epsilon\beta)}
    (1 + C \epsilon \log(1/\epsilon)^4).
\]
\end{proposition}
\begin{proof}
To find an upper bound, we split the set $A_\epsilon^\complement$ into three parts
(\Cref{fig:well upper schematic}):
the bottom of the well $B_\epsilon$, the set of mildly oscillating functions
\begin{equation}
D_1 \coloneqq \left\{ u \in L^2(\Tor) \colon \rho < |\hat u(0)| < 2\pi/\beta - \rho,
    u_\perp \in B_\epsilon^\perp \right\},
\end{equation}
where $\rho \coloneqq \sqrt\epsilon \log(1/\epsilon)$,
and the remaining set
\begin{equation}
D_2 \coloneqq \left\{ u \in L^2(\Tor) \colon |\hat u(0)| < 2\pi/\beta - \rho,
    u_\perp \in (B_\epsilon^\perp)^\complement \right\}.
\end{equation}
In all of these sets we can use the trivial bound $h_{B,A}(u) \leq 1$.

\bigskip\noindent
In the set $B_\epsilon$, we can use the Gaussian expansion from \Cref{thm:easy potential decomposition}.
The dilated set $\Blog$ is contained in an $L^\infty$-ball of radius $R \simeq \log(1/\epsilon)$.
As we put $u = \sqrt\epsilon v$ into the remainder term, we get
\begin{equation}
\exp\!\left(\frac 1 \epsilon \int_\Tor R(u(x)) \dx \right)
\leq \exp\!\left(C \epsilon \int_\Tor v(x)^4 \dx \right)
%\leq \exp(C\epsilon R^4)
\leq 1 + C \epsilon \log(1/\epsilon)^4.
\end{equation}
Together with the trivial bound for the Gaussian integral this implies
\begin{equation}\label{eq:well concentration upper}
\int_{B_\epsilon} e^{-F(u)/\epsilon} \du
\leq e^{2\pi\gamma/(\epsilon\beta)}
    \sqrt{\frac{\epsilon}{\gamma\beta} \prod_{n \neq 0} \frac{2\pi\epsilon}{n^2 + \gamma\beta}}
    (1 + C\epsilon \log(1/\epsilon)^4).
\end{equation}

\bigskip\noindent
In the set $D_1$,
we have $\norm{u_\perp}_\infty \leq \rho/2$ once $\epsilon$ is small enough.
This is because we defined $\bar u$ and $u_\perp$ to decay at different rates.
As $\cos(\beta u)$ takes its maximum at $u = 2k\pi/\beta$, we can do a rough estimate like
\begin{equation}
\frac{\gamma}{\epsilon\beta} \cos(\beta u(x))
\leq \frac{\gamma}{\epsilon\beta} - \frac{C u(x)^2}{\epsilon}
\leq \frac{\gamma}{\epsilon\beta} - C \log(1/\epsilon)^2
\end{equation}
in the potential.
The $\abs{\nabla u}^2$ term yields a Gaussian measure over the oscillatory modes,
and the mean is taken over a bounded set, so
\begin{equation}
\begin{split}
\int_{D_1} e^{-F(u)/\epsilon} \du
&\leq e^{\gamma/(\epsilon\beta) - C \log(1/\epsilon)^2} 2(2\pi/\beta - 2\delta)
    \sqrt{\prod_{0 < \abs n \leq N} \frac{2\pi\epsilon}{n^2}} \int_{\R^{2N}} \hspace{-0.5em} \diff\rho\\
&\lesssim \epsilon \cdot e^{\gamma/(\epsilon\beta)}
    \sqrt{\prod_{\abs n \leq N} \frac{2\pi\epsilon}{n^2 + \gamma\beta}}.
\end{split}
\end{equation}
The products on the last two lines are comparable by \Cref{sec:prefactor}.

\bigskip\noindent
In the set $D_2$ we have no better bound than $\gamma/(\epsilon\beta)$ for the cosine term,
but the integral over $(B_\epsilon^\perp)^\complement$ is negligible.
The second case of \Cref{thm:well concentration} applies,
so also this term is of order $\epsilon e^{\gamma/(\epsilon\beta)}$.
\end{proof}

\subsection{Constant saddle (\texorpdfstring{$\gamma\beta < 1$}{γβ < 1})}\label{sec:easy capacity}

In this case the Hessian matrix of $F(\pm \pi/\beta)$ is diagonal with respect to the ambient Fourier basis,
so all eigenvalues are explicit.
The computations are direct adaptations of those in \cite{berglund_sharp_2013}
or \cite[Sections~4.1 and~5.3]{berglund_eyring_2017}.

\begin{proposition}[Capacity, upper bound]\label{thm:easy capacity upper}
For $\epsilon$ small enough, there is the uniform-in-$N$ estimate
\[
\capa(A_\epsilon, B_\epsilon) \leq 2 \sqrt{\epsilon\gamma\beta
        \prod_{0 < \abs n \leq N} \frac{2\pi \epsilon}{n^2 - \gamma\beta}}
    e^{-2\pi\gamma/(\epsilon\beta)}
    (1 + C\epsilon).
\]
\end{proposition}
\begin{proof}
Since $F(u)$ is symmetric with respect to $\hat u(0)$ and
\begin{equation}
\capa(A_\epsilon, B_\epsilon) \leq \capa(A_\epsilon^-, B_\epsilon) + \capa(A_\epsilon^+, B_\epsilon),
\end{equation}
it is sufficient to bound $2 \capa(A_\epsilon^+, B_\epsilon)$.
For an upper bound we need to choose a suitable test function.
Let us again abbreviate $\rho \coloneqq \sqrt{\epsilon} \log(1/\epsilon)$ and put
\begin{equation}
g(v) = \begin{cases}
0, &\quad \hat v(0) < \pi/\beta - \rho,\\
\displaystyle \frac{\int_{-\rho}^{\hat v(0) - \pi/\beta} e^{-2\pi \gamma\beta t^2/2\epsilon} \dt}
    {\int_{-\rho}^{\rho} e^{-2\pi \gamma\beta t^2/2\epsilon} \dt},
    &\quad\text{otherwise},\\
1, &\quad \hat v(0) > \pi/\beta + \rho.
\end{cases}
\end{equation}
The denominator evaluates, once $\epsilon$ is small enough, to
\begin{equation}
\sqrt\epsilon \int_{-\log(1/\epsilon)}^{\log(1/\epsilon)} \hspace{-0.5em}
    \exp\!\left(-\frac{2\pi \gamma\beta}{2} t^2\right) \dt
\geq \sqrt{\frac{\epsilon}{\gamma\beta}} (1 - C\epsilon).
\end{equation}
Hence
\begin{multline}\label{eq:easy capacity upper expanded}
2\capa(A_\epsilon^+, B_\epsilon)
\leq 2 \cdot \gamma\beta (1 + C\epsilon)\\
    \cdot \int_{-\rho}^{\rho} \diff{\bar v} \int_{\R^{2N}} \hspace{-1em} \diff v_\perp \exp\!\left( -\frac 1 \epsilon \int_{\Tor}
        \frac{\abs{\nabla v}^2}{2} - \frac\gamma\beta \cos(\pi + \beta v) + \gamma\beta \bar v^2 \dx \right).
\end{multline}
We can again use the expansion from \Cref{thm:easy potential decomposition}.
Note that the sign of the cosine is now negative,
and the remainder $R(u)$ can now be ignored as non-positive.
This leaves us with the Gaussian integral
\begin{equation}\label{eq:easy capacity upper integrals}
e^{-2\pi\gamma/(\epsilon\beta)}
\int_{-\rho}^{\rho} \diff{\bar v}
    \exp\!\left( - \frac{2\pi \gamma\beta}{2\epsilon} \bar v^2 \dx \right)
\int_{\R^{2N}} \hspace{-1em} \diff v_\perp \exp\!\left( -\int_{\Tor}
    \frac{\abs{\nabla v}^2 - \gamma\beta v_\perp^2}{2\epsilon} \dx \right).
\end{equation}
Since $\gamma\beta < 1$, the latter integral can be normalized to be Gaussian.
Hence
\begin{equation}
\eqref{eq:easy capacity upper integrals}
\leq e^{-2\pi\gamma/(\epsilon\beta)} \sqrt{\frac{\epsilon}{\gamma\beta}} (1 + C\epsilon)
    \cdot \sqrt{\prod_{0 < \abs n \leq N} \frac{2\pi \epsilon}{n^2 - \gamma\beta}},
\end{equation}
which together with~\eqref{eq:easy capacity upper expanded} leads to the stated bound.
\end{proof}

\begin{proposition}[Capacity, lower bound]\label{thm:easy capacity lower}
As above,
\[
\capa(A_\epsilon, B_\epsilon) \geq 2 \sqrt{\epsilon\gamma\beta
        \prod_{0 < \abs n \leq N} \frac{2\pi \epsilon}{n^2 - \gamma\beta}}
    e^{-2\pi\gamma/(\epsilon\beta)}
    (1 - C\epsilon).
\]
\end{proposition}
\begin{proof}
Let us estimate the gradient in \Cref{def:capacity} from below by
\begin{equation}
\begin{split}
\capa(A_\epsilon, B_\epsilon)
&\geq \inf_{g} \epsilon \int_{\R^{2N+1}} \abs{\partial_{\bar u} g(u)}^2 \exp(-F(u)/\epsilon) \du\\
&\geq \epsilon \int_{\R^{2N}} \inf_{\tilde g}
    \int_{\R} \abs{\partial_{\bar u} \tilde g(\bar u)}^2 \exp(-F(\bar u + u_\perp)/\epsilon)
    \diff{\bar u} \du_\perp.
\end{split}
\end{equation}
In the latter equation we optimize over $\tilde g \colon \R \to \R$ such that $\tilde g(\bar u) = 0$
for $|\bar u| < \rho$
and $\tilde g(\bar u) \geq 1$ for $|\bar u| > 2\pi/\beta - \rho$,
where again $\rho \coloneqq \sqrt{\epsilon} \log(1/\epsilon)$.
Now that the problem is reduced to one dimension,
we can again use symmetry to only consider the case $\bar u > 0$.
We can also weaken the assumption to $\tilde g(0) = 0$ and $\tilde g(2\pi/\beta) \geq 1$.

By the Euler--Lagrange equation, the minimizer $\tilde g$ satisfies
\begin{equation}
\tilde g'(x) = \frac{e^{F(x + u_\perp)/\epsilon}}
    {\int_{0}^{2\pi/\beta} e^{F(y + u_\perp)/\epsilon} \dy },
\quad 0 \leq x \leq \frac{2\pi}{\beta}.
\end{equation}
After translating $\bar u$ to be centered at the saddle, we are left with
\begin{equation}
\capa(A_\epsilon, B_\epsilon)
\geq 2 \epsilon \int_{\R^{2N}}
    \left[\int_{-\pi/\beta}^{\pi/\beta} e^{F(\pi/\beta + \bar u + u_\perp)/\epsilon} \diff{\bar u}\right]^{-1} \du_\perp.
\end{equation}
Since we are interested in a lower bound, we may restrict to $u_\perp \in B_\epsilon^\perp$.
Therefore we are back in the first two cases of \Cref{thm:well upper}.

When $|\bar u|$ is at most $\sqrt{\epsilon}\log(1/\epsilon)$, we have
\begin{equation}
\frac{F}{\epsilon}
= \frac{2\pi\gamma}{\epsilon\beta} - \frac{2\pi \gamma\beta \bar u^2}{\epsilon}
    + \frac{1}{2\epsilon} \sum_{0 < \abs n \leq N} (n^2 - \gamma\beta) \hat u(n)^2
    + R(\bar u + u_\perp),
\end{equation}
where $R$ is non-negative and of order $\epsilon \log(1/\epsilon)^2$.
This region gives the dominant contribution,
which factorizes neatly into Gaussian densities for mean and oscillating parts.

When $|\bar u|$ is larger than $\sqrt{\epsilon} \log(1/\epsilon)$,
we use the assumption $u_\perp \in B_\epsilon^\perp$ to control the cosine with $\bar u$,
and add and subtract the $\gamma\beta$ term:
\begin{equation}
\frac{F}{\epsilon}
\leq \frac{2\pi\gamma}{\epsilon\beta} - \frac{C \bar u^2}{\epsilon}
    + \frac{1}{2\epsilon} \sum_{0 < \abs n \leq N} (n^2 - \gamma\beta) \hat u(n)^2
    + \frac{\gamma\beta}{2\epsilon} \sum_{0 < \abs n \leq N} \hat u(n)^2.
\end{equation}
By assumption the second term is smaller than $-C \log(\epsilon)^2$,
whereas the last term is smaller than $c_0 \log(\epsilon)$
by the $\Holder^{1/2-} \hookrightarrow L^2$ embedding.
Therefore this contribution is vanishing of order $\epsilon$, and
\begin{equation}
\int_{-\pi/\beta}^{\pi/\beta} \hspace{-0.5em} e^{F(\pi/\beta + \bar u + u_\perp)/\epsilon} \diff{\bar u}
\leq \exp\!\left( \frac{2\pi\gamma}{\epsilon\beta}
        + \sum_{n \neq 0} \frac{(n^2 - \gamma\beta) \hat u(n)^2}{2\epsilon} \right)
    \sqrt{\frac{\epsilon}{\gamma\beta}} (1 + C\epsilon).
\end{equation}
Now it remains to compute the Gaussian integral over the oscillating modes,
which again uses \Cref{thm:well concentration}.
\end{proof}

\subsection{Non-constant saddle (\texorpdfstring{$\gamma\beta > 1$}{γβ > 1})}
\label{sec:hard capacity}

In this case the eigenvalues of $\Hess F[u_\ast]$ at the saddle $u_\ast$ are not explicitly known,
and the one-dimensional saddle manifold yields an extra prefactor.
We still follow the approach of Berglund and Gentz \cite[Sections~6.1.3 and~6.2]{berglund_sharp_2013}.

The lower bound for capacity is bootstrapped
from a large-deviation estimate derived in \Cref{sec:ldp}.
The estimate does not give enough control if the sets vanish in $\epsilon \to 0$,
so we need to keep the sets fixed.
Therefore we now define the initial data to be sampled from
\begin{equation}
B \coloneqq \{ u \in L^2 \colon \norm{u}_{\Holder^{1/2-\kappa}} < \delta \},
\end{equation}
and consider the hitting time of
\begin{equation}
A \coloneqq \{ u \in L^2 \colon \norm{u - 2\pi/\beta}_{\Holder^{1/2-\kappa}} < \delta \}
    \cup \{ u \in L^2 \colon \norm{u + 2\pi/\beta}_{\Holder^{1/2-\kappa}} < \delta \}.
\end{equation}
The small parameter $\delta > 0$ is fixed in \Cref{thm:hard capacity lower}.
The estimates on the potential well done in \Cref{sec:well} apply also here.

Let $\manif_+$ and $\manif_-$ be the one-dimensional manifolds corresponding to the essential gates
at $\hat u(0) = \pm \pi/\beta$.
Given any point $u_\ast \in \manif_\pm$,
the others are given by $\mathcal T_t u_\ast(x) = u_\ast(x-t)$.
In Fourier space the action can be written as
\begin{equation}
\widehat{\mathcal T_t v}(n) =
\cos(nt) \hat v(n) - \sin(nt) \hat v(-n).
\end{equation}

\begin{lemma}[Length of saddle manifold]\label{thm:hard capacity manifold length}
Let us fix one point $u_\ast \in \manif$.
The length of the curve $\manif$ is then $\ell = 2\pi \norm{\partial_x u_\ast}_{L^2(\Tor)}$.
\end{lemma}
\begin{proof}
It follows from the definition that
\begin{equation}
\partial_t \widehat{\mathcal T_t v}(n) =
    -n \left[ \sin(nt) \hat v(n) + \cos(nt) \hat v(-n) \right].
\end{equation}
By the symmetry of the sum, we then have
\begin{equation}
\sum_{\abs n \leq N} [\partial_t \widehat{\mathcal T_t v}(n)]^2
= \sum_{\abs n \leq N} n^2 (\sin(nt)^2 + \cos(nt)^2) \hat v(n)^2.
\end{equation}
Hence this parametrization yields
\begin{equation}
\int_\manif \diff u
= \int_0^{2\pi} \norm{\partial_t \mathcal T_t u_\ast}_{\R^{2N+1}} \dt
= 2\pi \norm{\partial_x u_\ast}_{L^2(\Tor)}.
\qedhere
\end{equation}
\end{proof}

For any $u \in \manif_\pm$, the matrix $\Hess F[u]$ has an orthonormal eigenbasis
($\check e_{-}$, $\check e_0$, \ldots, $\check e_{2N-1}$) with
eigenvalues ($-\mu$, $0$, $\lambda_1$, $\ldots$, $\lambda_{2N-1}$).
The eigenvector $e_0$ is tangent to $\manif_\pm$ at $u$.
The eigenvalues are independent of the point $u$
whereas the eigenvectors $e_i$ are also related by the translation $\mathcal T_t$.

Any $v$ in the normal subspace of $\manif$ at $u$ can be thus written
with the eigenbasis decomposition
\begin{equation}
v(x) = u(x) + \check v(-1) \check e_{-}(x) + \sum_{n = 1}^{2N-1} \check v(n) \check e_n(x).
\end{equation}
We will denote the sum term by $v_\dagger(x)$.
We use this decomposition to define neighbourhoods of the essential gate in the lemmas below
-- the precise form is case-dependent.
Although the basis depends on the base point $u \in \manif$,
we will suppress this dependency whenever it does not matter.

\begin{remark}\label{rem:hard capacity projection}
There is a technical nuance in Step~2 and \Cref{thm:hard capacity upper error} below.
We could use \Cref{thm:hard capacity stable manifold} without $M$-truncation to estimate the potential
by $V[u] \geq V[u_\ast] + m \dist(u, \manif)_{H^1}$.
This would however introduce a factor of $m^{-N}$ into the error term,
which would blow up since $m$ would be small.
This appears to be an oversight in \cite[Eq.~(6.25)]{berglund_sharp_2013}.

Since the Galerkin approximations of stationary solutions converge rapidly,
we can introduce an auxiliary $\epsilon$-dependent truncation of the stable manifold.
This replaces the factor $m^{-N}$ by $\exp(\epsilon^{-C})$,
where $C$ is small enough for the term to be controllable.

This detour makes the error estimate in \Cref{thm:hard capacity upper}
worse than the $\sqrt\epsilon \log(1/\epsilon)^C$ in \Cref{thm:hard capacity lower}.
We do not aim for an optimal rate.
\end{remark}

\begin{proposition}[Upper bound for capacity]\label{thm:hard capacity upper}
We have
\[
\capa(A, B) \leq 2\ell
    \sqrt{\frac{\epsilon\mu}{2\pi} \prod_{n \geq 1} \frac{2\pi\epsilon}{\lambda_n}}
     e^{-F(u_\ast)/\epsilon} (1 + C \epsilon^{1/4} \log(1/\epsilon)^3)
\]
when $N \geq N_0 \epsilon^{-1/6}$.
\end{proposition}
\begin{proof}
We can use symmetry to consider only one saddle manifold,
like we did in \Cref{thm:easy capacity upper}.
We present this proof in two parts:
First we estimate the contribution near $\manif$,
and then show the remaining contribution to be negligible.

The neighbourhood we choose is
$\neigh \coloneqq \{ u + v \colon u \in \manif, v \in D_u \}$, where
\begin{equation}
\begin{gathered}
D_u \coloneqq \big\{ v \in L^2(\Tor) \colon \abs{\check v(-1)} < \epsilon^{5/12},
    v_\dagger \in D_u^\dagger \big\}, \text{ with}\\
D_u^\dagger \coloneqq \big\{
    v_\dagger \in L^2(\Tor) \colon
    \norm{v_\dagger}_{\Holder^{1/2-}} < \epsilon^{5/12} \log(1/\epsilon)
\big\}.
\end{gathered}
\end{equation}
Note that now we control the unstable mode with the stable modes,
as opposed to the earlier setting.
In the integration we will essentially replace $\neigh$ by $\manif \times D_{u_\ast}$,
since possible overlaps between $D_{u}$ and $D_{u'}$ for some $u$, $u' \in \manif$
can be safely overcounted.

\bigskip\noindent\emph{Step 1: The bulk.}
This argument is parallel to \Cref{thm:easy capacity upper}.
For a fixed $u_\ast \in \manif$, we define the test function
\begin{equation}
g_\ast(v) = \begin{cases}
    0, &\quad \check v(-1) < -\rho,\\
    \displaystyle \frac{\int_{-\rho}^{\check v(-1)} e^{-\mu t^2/2\epsilon} \dt}
        {\int_{-\rho}^{\rho} e^{-2\pi \mu t^2/2\epsilon} \dt},
        &\quad\text{otherwise},\\
    1, &\quad \check v(-1) > \rho.
\end{cases}
\end{equation}
where $\rho = \epsilon^{5/12}$.
We extend this to an $H^1$ function on $\neigh$ by $g(u_t, v) \coloneqq g_\ast(\mathcal T_{-t} v)$.
The denominator is still larger than $\sqrt{2\pi\epsilon/\mu} (1 - C\epsilon)$,
once $\epsilon$ is small enough.

As we plug $g$ into the definition of capacity, we get
\begin{equation}\label{eq:hard capacity upper plugged}
\epsilon \int_\manif \int_{D_u} \frac{\mu}{2\pi\epsilon} (1-C\epsilon)^2
    e^{-F(u+v)/\epsilon - 2\mu \check v(-1)^2/2\epsilon} \diff v \diff u.
\end{equation}
As the inner integral is independent of $u \in \manif$,
the outer integral yields a factor of $\ell$.
We then use the decomposition \eqref{eq:potential decomposition generic}
and the $L^3$ bound from $D^\dagger$ to bound
\begin{equation}
\begin{split}
&\mathrel{\phantom{=}} \int_{D_{u_\ast}}
    \exp\!\left(-F(u_\ast+v)/\epsilon - 2\mu \check v(-1)^2/2\epsilon \right) \diff v\\
&= \int_{D_{u_\ast}} \exp\!\left(
    -\frac{F(u_\ast)}{\epsilon} - \frac{\mu \check v(-1)^2}{2\epsilon}
    - \sum_{n \geq 1} \frac{\lambda_n \check v(n)^2}{2\epsilon}
    + \frac 1 \epsilon \int_\Tor \bigO(v^3) \dx
\right) \diff v\\
&\leq \sqrt{\frac{2\pi\epsilon}{\mu} \prod_{n \geq 1} \frac{2\pi\epsilon}{\lambda_n}}
    e^{-F(u_\ast)/\epsilon}
    (1 + C \epsilon^{1/4} \log(1/\epsilon)^3).
\end{split}
\end{equation}
Together with \eqref{eq:hard capacity upper plugged} this gives the upper bound
\begin{equation}\label{eq:hard capacity bulk}
\ell e^{-F(u_\ast)/\epsilon}
    \sqrt{\frac{\epsilon\mu}{2\pi} \prod_{n \geq 1} \frac{2\pi\epsilon}{\lambda_n}}
    (1 + C \epsilon^{1/4} \log(1/\epsilon)^3)
\end{equation}
for the integral over $\neigh$.

\bigskip\noindent\emph{Step 2: The remainder.}
We then need to show that $g$ can be extended to the set $\R^{2N+1} \setminus \neigh$
in a way that gives a negligible contribution.
Because the function $g$ only depends on one coordinate,
it can be extended to a $\rho$-neighbourhood of $\manif$ in $L^2$ norm.

We now divide $\R^{2N+1}$ into two components separated by an union of stable manifolds.
As explained in \Cref{rem:hard capacity projection},
we introduce an auxiliary frequency truncation for this step.

Let $M = \epsilon^{-1/6}$, rounded to an integer.
Denote by $\manif_M$ the saddle manifold of the $M$-truncated equation~\eqref{eq:DSG truncated},
and by $u_{\ast}^{(M)}$ a representative of that saddle.
Further, define
\begin{equation}
\mathcal S^{(M)} \coloneqq \left\{ u \in L^2(\Tor) \colon
    \projM u \text{ belongs to the combined stable manifold} \right\}\!.
\end{equation}
Here we take the union of stable manifolds of \emph{all} mean-$\pi/\beta$ stationary solutions
to the $M$-truncated equation.
It then follows that $\mathcal S^{(M)}$ separates $\R^{2N+1}$ into two components once $N \geq M$.

Let $\mathcal S_\epsilon$ be a layer of thickness $\rho$ around $\mathcal S^{(M)}$,
and $\mathcal A$ the domain of attraction of the potential well $A$ in $\R^{2N+1}$.
Let $\theta \colon \R \to \R_{\geq 0}$ be a smooth, non-increasing function
such that $\theta(r) = 1$ for $r \leq 0$ and $\theta(r) = 0$ for $r \geq 1$.
We then define
\begin{equation}
\tilde g(u) =
\left\{
\begin{aligned}
    g(u),\quad &\text{when } \dist(u, \manif)_{L^2} \leq \rho,\\
    \theta(\rho^{-1} \dist(u, \mathcal A \setminus \mathcal S_\epsilon)),\quad
        &\text{when } \dist(u, \manif)_{L^2} \geq 2\rho,
\end{aligned}
\right.
\end{equation}
smoothly interpolated in the intermediate region.
It then follows that $\abs{\nabla \tilde g}^2$ is of order at most $\epsilon^{-5/6}$
and vanishes outside $\mathcal S_\epsilon$.
Thus
\begin{equation}
\capa(A, B) \leq 2 \cdot \eqref{eq:hard capacity bulk}
    + \epsilon\, \bigO(\epsilon^{-5/6}) \int_{\mathcal S_\epsilon \setminus \neigh}
        e^{-F(u)/\epsilon} \du.
\end{equation}
We are left to do a concentration of measure argument,
which we separate as the two lemmas below.
\end{proof}

\begin{lemma}[Estimate on stable manifold]\label{thm:hard capacity stable manifold}
There exists $m_1 > 0$ independent of $M$
such that for all $\tilde u \in \mathcal S^{(M)}$ we have
\[
F(\tilde u) \geq F(u_{\ast}^{(M)})
    + \inf_{u_{\ast}^{(M)} \in \manif_M} m_1 \norm{u_{\ast}^{(M)} - \tilde u}_{H^1}^2,
\]
\end{lemma}
\begin{proof}
Fix $u_{\ast}^{(M)} \in \manif_M$ and let $v \in L^2(\Tor)$ satisfy $\projMperp v = 0$.
By~\eqref{eq:potential decomposition generic} we have for $r_0$ small enough the bound
\begin{equation}
F(u_{\ast}^{(M)} + v) \geq F(u_{\ast}^{(M)}) + \frac 1 4 (\cova_M v, v),
\quad \norm{v}_{H^1} \leq r_0.
\end{equation}
Since the positive eigenvalues of $\cova_M$ are bounded away from $0$ uniformly in $M$,
and asymptotically equal to those of $-\Laplace + 1$,
this implies
\begin{equation}\label{eq:hard capacity error small ball}
F(u_{\ast}^{(M)} + v) \geq F(u_{\ast}^{(M)}) + m_0 \norm{v}_{H^1}^2,
    \quad \norm{v}_{H^1} \leq r_0.
\end{equation}

On the other hand we can use the trivial lower bound and triangle inequality to estimate
\begin{equation}
\begin{split}
F(u_{\ast}^{(M)} + v)
&\geq \frac 1 2 \norm{\nabla (u_{\ast}^{(M)} + v)}_{L^2}^2 - \frac\gamma\beta\\
&\geq \frac 1 4 \norm{\nabla v}_{L^2}^2 - \frac 1 2 \norm{\nabla u_{\ast}^{(M)}}_{L^2}^2
    - \frac \gamma \beta.
\end{split}
\end{equation}
If we have
\begin{equation}
\frac 1 8 \norm{\nabla v}_{L^2}^2
\geq \frac 1 2 \norm{\nabla u_{\ast}^{(M)}}_{L^2}^2 + F(u_{\ast}^{(M)}) + \frac\gamma\beta,
\end{equation}
then this implies
\begin{equation}
F(u_{\ast}^{(M)} + v)
\geq F(u_{\ast}^{(M)}) + \frac 1 8 \norm{\nabla v}_{L^2}^2.
\end{equation}
The right-hand side is then upgraded to an $H^1$ bound by
incorporating also the confining term~\eqref{eq:confining potential}, since
\begin{equation}
\max(0, |\hat u_{\ast}^{(M)}(0) + \hat v(0)| -k2\pi/\beta)^2
\geq C \hat v(0)^2
\end{equation}
when $\hat v(0)$ is large enough.

An interpolation argument as in \cite[Lemma~6.1]{berglund_sharp_2013}
combines these two bounds when $v$ is of intermediate size.
\end{proof}

\begin{lemma}[Concentration on $\neigh$]\label{thm:hard capacity upper error}
With $\mathcal S_\epsilon$ and $\neigh$ defined as above, we have
\[
\int_{\mathcal S_\epsilon \setminus \neigh}
    e^{-F(u)/\epsilon} \du
\leq C e^{-C / \epsilon^{1/6}} e^{-F(u_\ast)/\epsilon}
    \sqrt{\frac{\epsilon\mu}{2\pi} \prod_{n \geq 1} \frac{2\pi\epsilon}{\lambda_n}}.
\]
\end{lemma}
\begin{proof}
A general $u \in \mathcal S_\epsilon$ can be written as
$\tilde u + u_{\mathrm{us}} + u_\perp$,
where $\tilde u \in \mathcal S_M$, $\projM u_\perp = 0$, and
$u_{\mathrm{us}}$ is normal to $\mathcal S_M$ at $\tilde u$ in the $(2M+1)$-dimensional space.
Since $\nabla_{\mathrm{us}} F = 0$ on $\mathcal S_M$
and the Hessian of $F$ is bounded from below by $-2\pi\gamma/\beta$, we have
\begin{equation}
F(\tilde u + u_{\mathrm{us}}) \geq F(u_{\ast}^{(M)})
+ \inf_{u_{\ast}^{(M)} \in \manif_M} m_1 \norm{\nabla (u_{\ast}^{(M)} - \tilde u)}_{L^2}^2
- C \rho^2.
\end{equation}
To estimate the effect of $u_\perp$ on the nonlinear term, we begin with the decomposition
\begin{equation}
\begin{split}
&\mathrel{\phantom{=}} \cos(\beta[\tilde u + u_{\mathrm{us}} + u_\perp])\\
&= \cos(\beta[\tilde u + u_{\mathrm{us}}])\cos(\beta u_\perp)
    - \sin(\beta[\tilde u + u_{\mathrm{us}}])\sin(\beta u_\perp)\\
&\geq \cos(\beta[\tilde u + u_{\mathrm{us}}]) - C u_\perp^2
    - \sin(\beta[\tilde u + u_{\mathrm{us}}])\sin(\beta u_\perp).
\end{split}
\end{equation}
The last term is further split into
\begin{align}
&\mathrel{\phantom{=}} \sin(\beta[\tilde u + u_{\mathrm{us}}])\sin(\beta u_\perp) \notag\\
&= \sin(\beta u_{\ast}^{(M)}) \sin(\beta u_\perp)
    + [\sin(\beta[\tilde u + u_{\mathrm{us}}]) - \sin(\beta u_{\ast}^{(M)})] \sin(\beta u_\perp)\\
&= -\frac 1 \gamma (\Laplace u_{\ast}^{(M)}) [\beta u_\perp + \bigO(u_\perp^2)]
    + [\sin(\beta[\tilde u + u_{\mathrm{us}}]) - \sin(\beta u_{\ast}^{(M)})] \sin(\beta u_\perp) \notag.
\end{align}
Here $(\Laplace u_{\ast}^{(M)}) u_\perp$ vanishes in the integration by orthogonality.
The other terms are bounded by
\begin{equation}
\frac{C_1}{\delta} u_\perp^2 + \delta \abs{[\tilde u + u_{\mathrm{us}}] - u_{\ast}^{(M)}}^2.
\end{equation}
By choosing $\delta = m_1 / 2$, we then have
\begin{equation}
\begin{split}
F(\tilde u + u_{\mathrm{us}} + u_\perp)
&\geq F(u_\ast)
    + \frac{m_1}{2} \inf_{u_{\ast}^{(M)}} \norm{u_{\ast}^{(M)} - \tilde u}_{H^1}^2
    - C \rho^2\\
    &\qquad + \frac{\norm{\nabla u_\perp}_{L^2}^2}{2} - C_1 \norm{u_\perp}_{L^2}^2
    - C \smallnorm{u_\ast - u_{\ast}^{(M)}}_{H^1}.
\end{split}
\end{equation}
We assume $\epsilon$ to be small enough that $M^2 - 2C_1 > 0$,
ensuring that we get a well-defined Gaussian for $u_\perp$.
The last term comes from approximating $F(u_\ast)$ by $F(u_{\ast}^{(M)})$;
by \Cref{thm:critical galerkin} it is of arbitrary polynomial order in $M$, say $\epsilon^2$.

In computing the integral, we again use Fubini to factorize it as
\begin{equation}\label{eq:hard capacity upper almost there}
\begin{split}
&\mathrel{\phantom{=}} \ell \int_{\R^{2N}} \I_{\mathcal S_\epsilon \setminus \neigh}
    \exp\!\left(-\frac{F(u_\ast + v)}{\epsilon}\right) \diff v\\
&\lesssim e^{-F(u_\ast)/\epsilon + C \epsilon^{-1/6} + C\epsilon}\\
    &\quad\qquad \sqrt{\prod_{0 < \abs k \leq M} \frac{2\pi\epsilon}{m_1 (1 + k^2)}
        \prod_{M < \abs j \leq N} \frac{2\pi\epsilon}{j^2 - 2C_1}}
    \int_{\R^{2N}} \I_{\mathcal S_\epsilon \setminus \neigh} \diff \tilde\rho(v),
\end{split}
\end{equation}
where $\tilde\rho$ is a Gaussian with density
\begin{equation}
v \mapsto \exp\!\left( -\frac{1}{2\epsilon} \sum_{\abs k \leq M} m_1 (1+k^2) \hat v(k)^2
    - \frac{1}{2\epsilon} \sum_{\abs j > M} (j^2 - C) \hat v(j)^2 \right).
\end{equation}
The square root term can be bounded with eigenvalue asymptotics as
\begin{equation}
C m_1^{-M} \sqrt{\prod_{n \geq 1} \frac{2\pi\epsilon}{\lambda_n}}
= C_1 \exp(C_2 \epsilon^{-1/6}) \sqrt{\prod_{n \geq 1} \frac{2\pi\epsilon}{\lambda_n}},
\end{equation}
where the constants are independent of $M$.
Adapting \Cref{thm:well concentration}, we see that the integral is bounded by
\begin{equation}
\exp\!\left( C_1 - \frac{C_2 \log(1/\epsilon)}{\epsilon^{1/6}} \right)\!.
\end{equation}
When $\epsilon$ is small enough,
this cancels the exponential terms in~\eqref{eq:hard capacity upper almost there}.
\end{proof}

The proof of the lower bound follows \cite[Proposition~6.3]{berglund_sharp_2013}:
we use the same variational principle as in \Cref{thm:easy capacity lower},
but now the boundary values for the test function come from a Fre\u{\i}dlin--Wentzell estimate.

\begin{proposition}[Lower bound for capacity]\label{thm:hard capacity lower}
We have
\[
\capa(A, B)
\geq 2\ell
    \sqrt{\frac{\epsilon\mu}{2\pi} \prod_{n \geq 1} \frac{2\pi \epsilon}{\lambda_n}}
    e^{-F(u_\ast)/\epsilon} (1 - C\sqrt{\epsilon}\log(1/\epsilon)^3),
\]
when $\epsilon < \epsilon_0$ and $N > N_0 \epsilon^{-3}$ for certain $\epsilon_0, N_0 > 0$.
\end{proposition}
\begin{proof}
We need to modify the neighbourhood slightly to use the large-deviation estimate.
We set $\neigh' \coloneqq \{ u + v \colon u \in \manif, v \in D_u' \}$, where
\begin{equation}
\begin{gathered}
D_u' \coloneqq \{ v \in L^2(\Tor) \colon \abs{\check v(-1)} < \delta,
    v_\dagger \in D_u^\dagger \},\\
D_u^\dagger \coloneqq \big\{
    v_\dagger \in L^2(\Tor) \colon
    \norm{v_\dagger}_{\Holder^{1/2-\kappa/2}} \leq \sqrt{c_0 \epsilon} \log(1/\epsilon)
\big\}.
\end{gathered}
\end{equation}
The set $D_u^\dagger$ can be taken to decay faster than in \Cref{thm:hard capacity upper}.

We use the characterization of the capacity with the equilibrium potential.
This permits us to again use the symmetry.
We have
\begin{equation}\label{eq:hard capacity lower variational}
\begin{split}
&\mathrel{\phantom{=}} \capa(A,B)\\
&\geq 2\epsilon \int_\manif \int_{D'_u}
    \abs{\partial_{\check v(-1)} h_{A,B}(u + v)}^2 \exp(-F(u+v)/\epsilon) \diff v \du\\
&= 2\epsilon \ell \int_{D^\dagger}
    \int_{-\delta}^\delta \abs{\partial_{\check v(-1)} h_{A,B}(u_\ast+v)}^2 \exp(-F(u_\ast+v)/\epsilon)
    \diff \check v(-1) \diff v_\dagger\\
&\geq 2\epsilon \ell \int_{D^\dagger}
    \inf_{\tilde g} \int_{-\delta}^\delta \tilde g'(t)^2 \exp(-F(u_\ast + v_\dagger + t \check e_{-})/\epsilon)
    \dt \diff v_\dagger.
\end{split}
\end{equation}
The infimum is taken over $\tilde g \in H^1(\R)$ such that
$\tilde g(\pm\delta) = h_{A,B}(u_\ast + v_\dagger \pm \delta\check e_{-})$.
We now use \Cref{thm:ldp equilibrium global} and the symmetry $h_{A,B} = 1 - h_{B,A}$
to estimate $h_{A,B}$ at the endpoints.
By the compact embedding of Hölder spaces,
the set $\neigh^\dagger_u$ is compact in $\Holder^{1/2-\kappa}$.
Therefore for any $\delta > 0$ there exist $\epsilon_0$ and $N_0$ such that
\begin{equation}
\tilde g(-\delta) < 4 \epsilon
\quad\text{and}\quad
\tilde g(\delta) > 1 - 4 \epsilon
\end{equation}
when $\epsilon < \epsilon_0$ and $N > N_0 \epsilon^{-3}$.

By the same Euler--Lagrange argument as in \Cref{thm:easy capacity lower},
the minimizer in \eqref{eq:hard capacity lower variational} then satisfies
\begin{equation}
\tilde g'(x) = \frac{(1-8\epsilon) e^{F(u_\ast + v_\dagger + x\check e_{-})/\epsilon}}
    {\int_{-\delta}^{\delta} e^{F(u_\ast + v_\dagger + y\check e_{-})/\epsilon} \dt },
\quad -\delta \leq x \leq \delta.
\end{equation}
This leaves us again to estimate
\begin{equation}\label{eq:hard capacity lower after el}
2 \epsilon \ell (1 - 8\epsilon) \int_{D^\dagger} \left[
    \int_{-\delta}^{\delta} e^{F(u_\ast + v_\dagger + y\check e_{-})/\epsilon} \dt
\right]^{-1} \diff v_\dagger.
\end{equation}
By \eqref{eq:potential decomposition generic} and orthogonality of the eigenspaces
we may extract the stable modes:
\begin{equation}
\begin{split}\label{eq:hard capacity lower extracted}
\eqref{eq:hard capacity lower after el}
&\geq 2 \epsilon \ell (1-8\epsilon) \int_{D^\dagger} \!\! \exp\!\left(
    -\frac{F(u_\ast)}{\epsilon} - \frac{1}{2\epsilon} (\cova v_\dagger, v_\dagger)
        - \frac 1 \epsilon \int_\Tor \bigO(v_\dagger^3) \dx \right)\\
&\hspace{8em}\left[
    \int_{-\delta}^{\delta} \exp\!\left( -\frac{\mu t^2}{2\epsilon} + \frac 1 \epsilon \bigO(t^3) \right) \dt
\right]^{-1} \diff v_\dagger.
\end{split}
\end{equation}
Since $\bigO(v_\dagger^3) \lesssim \epsilon^{3/2} \log(1/\epsilon)^3$,
the stable modes can again be estimated with \Cref{thm:well concentration}.
The error is of order $1 - \sqrt\epsilon \log(1/\epsilon)^3$.

For the remaining bracket we repeat the argument of \Cref{thm:easy capacity lower}.
Let $\rho_\delta = \delta \sqrt\epsilon \log(1/\epsilon)$.
We bound the integral as
\begin{equation}
\leq \int_{-\rho_\delta}^{\rho_\delta}
    \exp\!\left( -\frac{\mu t^2}{2\epsilon} + \frac 1 \epsilon \bigO(t^3) \right) \dt
+ 2 \int_{\rho_\delta}^{\infty}
    \exp\!\left( -\frac{(\mu - C\delta) t^2}{2\epsilon} \right) \dt.
\end{equation}
In the first term the error is of order $\sqrt{\epsilon} \log(1/\epsilon)^3$.
For the second term we choose $\delta$ small enough that $\mu - C\delta > 0$.
Then we get by Gaussian concentration that it too is of order $\sqrt{\epsilon} \log(1/\epsilon)^3$.
Hence
\begin{equation}
\eqref{eq:hard capacity lower extracted}
\geq 2\epsilon\ell e^{-F(u_\ast)/\epsilon}
    \sqrt{\prod_{n \geq 1} \frac{2\pi \epsilon}{\lambda_n}}
    \sqrt{\frac{\mu}{2\pi\epsilon}} (1 - C\sqrt{\epsilon}\log(1/\epsilon)^3),
\end{equation}
which finishes the proof.
\end{proof}

\section{Ratio of determinants}\label{sec:prefactor}

The determinant of an operator like $-\Laplace - \gamma\beta$ is not well-defined on its own,
but it can be made sense of relative to another operator as in
\begin{equation}
\frac{\det(-\Laplace - \gamma\beta)}{\det(-\Laplace + \gamma\beta)}.
\end{equation}
These \emph{functional determinants} have been studied
by Gel'fand and Yaglom \cite{gelfand_integration_1960}
and Forman \cite{forman_functional_1987}.

In the $\gamma\beta > 1$ case, we are further interested in the ratio with the zero eigenvalue removed.
Such a modified determinant was studied by McKane and Tarlie \cite{mckane_regularization_1995}.
In this section we follow their computational recipe.

Let us first introduce the method with the simpler $\gamma\beta < 1$ case.
As the eigenvalues of $-\Laplace \pm \gamma\beta$ have explicit expressions,
the ratio could also be computed with product formulas (see \cite[Eq.~2.56]{berglund_sharp_2013}).

\begin{theorem}[{{\cite[Section~2]{mckane_regularization_1995}}}]
Let us consider the operators $L = \Laplace + P(x)$ and $\hat L = \Laplace + \hat P(x)$,
where $P$ and $\hat P$ are real-valued functions.
Write the boundary conditions on $[0, 2\pi]$ with matrices $M$ and $N$:
\[
M \begin{pmatrix} u(0) \\ u'(0) \end{pmatrix}
+ N \begin{pmatrix} u(2\pi) \\ u'(2\pi) \end{pmatrix}
= \begin{pmatrix} 0 \\ 0 \end{pmatrix}.
\]
Let $y_i$ be two linearly independent solutions to $L y_i = 0$,
and define
\[
H(x) = \begin{pmatrix}
    y_1(x) & y_2(x)\\
    y_1'(x) & y_2'(x)
\end{pmatrix}.
\]
Repeat the same for $\hat L \hat y_i = 0$ to define $\hat H$.
Then
\[
\frac{\det(\Laplace + P(x))}{\det(\Laplace + \hat P(x))}
= \frac{\det(M + N H(2\pi) H(0)^{-1})}{\det(M + N \hat H(2\pi) \hat H(0)^{-1})}.
\]
\end{theorem}

For the periodic boundary condition we can choose
$M = \operatorname{Id}$ and $N = -\operatorname{Id}$.
Two independent solutions to $\Laplace y_i = -\gamma\beta y_i$
are given by $y_1(x) = \sin(\sqrt{\gamma\beta} x)$
and $y_2(x) = \cos(\sqrt{\gamma\beta} x)$.
It is then a straightforward computation to find
\begin{equation}
\det(M + N H(2\pi) H(0)^{-1})
= 4 \sin(\pi \sqrt{\gamma\beta})^2.
\end{equation}

The same computation can be carried out for $\Laplace \hat y_i = \gamma\beta$,
where $\sin$ is replaced by $\sinh$ throughout
and the resulting determinant has negative sign.
Therefore
\begin{equation}
\frac{\det(-\Laplace - \gamma\beta)}{\det(-\Laplace + \gamma\beta)}
= -\frac{\sin(\pi \sqrt{\gamma\beta})^2}{\sinh(\pi \sqrt{\gamma\beta})^2}.
\end{equation}
Together with the explicitly known negative eigenvalue,
this implies a simple expression for the prefactor in \Cref{thm:main transition time easy}.

\bigskip\noindent
For the $\gamma\beta > 1$ case,
we need a notion of determinant with the single zero eigenvalue removed.
We denote it by $\det'$.
McKane and Tarlie \cite{mckane_regularization_1995} adapt the argument by
modifying the boundary conditions and computing
\begin{equation}
\lim_{\epsilon \to 0} \frac{\det(M_\epsilon + N_\epsilon H(2\pi) H(0)^{-1})}{\lambda_\epsilon},
\end{equation}
where the eigenvalue $\lambda_\epsilon \to 0$.

\begin{theorem}[{{\cite[Section~5]{mckane_regularization_1995}}}]
With the previous notation,
\[
\frac{\det'(\Laplace + P(x))}{\det(\Laplace + \hat P(x))}
= \frac{(y_2(2\pi) - y_2(0)) \norm{y_1}_2^2}
    {y_1(0) \det H(0) \det(M + N \hat H(2\pi) \hat H(0)^{-1})}.
\]
\end{theorem}

The following argument follows \cite[Appendix]{mckane_regularization_1995}.
Let us recall from \eqref{eq:stationary solution} that the stationary solution is
\begin{equation}
u_\ast(x) = \frac{1}{\beta} \left[
    \pi + 2 \arcsin(\sqrt m \cd(\sqrt{\gamma\beta} x, m)) \right],
\quad m = \cos(\beta u(0)/2)^2.
\end{equation}
Here $m$ is chosen so that the period $4K(m)/\sqrt{\gamma\beta}$ is equal to $2\pi$.
It is convenient to keep $m$ in the notation throughout,
even though it is a function of $\sqrt{\gamma\beta}$.

We are interested in the operator $\cova = -\Laplace + \gamma\beta \cos(\beta u_\ast)$.
Two linearly independent solutions to $\cova y_i = 0$
are given by $y_1 = \partial_x u_\ast$ and $y_2 = \partial_m u_\ast$.
Both functions are $4K(m)/\sqrt{\gamma\beta}$-periodic and can be translated suitably.
In particular $y_1(0) = 0$, so we need to choose a different base point.

It is then a direct computation to verify
\begin{equation}
\begin{gathered}
y_1\!\left( \frac{K(m)}{\sqrt{\gamma\beta}} \right)
    = -2\sqrt{\frac{\gamma m}{\sqrt\beta}},\quad
y_1'\!\left( \frac{K(m)}{\sqrt{\gamma\beta}} \right) = 0,\\
y_2\!\left( \frac{K(m)}{\sqrt{\gamma\beta}} \right)
    = \frac{E(m) + (m-1) K(m)}{\beta (1-m) \sqrt m},\quad
y_2'\!\left( \frac{K(m)}{\sqrt{\gamma\beta}} \right)
    = -\sqrt{\frac{\gamma}{\beta m}},\\
y_2\!\left( \frac{5K(m)}{\sqrt{\gamma\beta}} \right)
    = 5 y_2\!\left( \frac{K(m)}{\sqrt{\gamma\beta}} \right),
\end{gathered}
\end{equation}
and the $L^2$ norm
\begin{equation}
\norm{y_1}_{L^2(\Tor)}^2 = \frac{16\gamma^{1/2}}{\beta^{3/2}}
    (E(m) - (1-m) K(m)).
\end{equation}
Here $E(m)$ is the complete elliptic integral of the second kind.
All together, these give
\begin{equation}\label{eq:determinant hard}
\frac{\det' \cova}{\det(-\Laplace + \gamma\beta)}
= -\frac{4 (E(m) + (m-1) K(m))^2}{\gamma\beta (1-m) m \sinh(\pi \sqrt{\gamma\beta})^2}.
\end{equation}
As the single negative eigenvalue is strictly nonzero and finite,
the prefactor in~\eqref{thm:main transition time hard} is well-defined.

\begin{remark}
Computing the negative eigenvalue $-\mu$ is an interesting question that we do not address here.
In the double-well case, asymptotics for it as the domain size $L \to \infty$
are given in \cite[Section~3.2.2]{rolland_computing_2016}.
\end{remark}

\appendix
\section{Large deviations}\label{sec:ldp}

We use the Fre\u{\i}dlin--Wentzell approach to bootstrap
a sharp capacity estimate in \Cref{thm:hard capacity lower}.
In this appendix we sketch the necessary results
without aiming for full generality or sharpness.

We begin with the large deviation principle inspired by \cite{faris_large_1982,freidlin_random_1988}.
Instead of using the topology of continuous functions $\Holder([0,T] \times \Tor)$,
we upgrade to the more natural $\Holder^{1/2-}(\Tor)$ topology in space.
Because the constants depend implicitly on the relevant sets,
we need to keep the sets fixed instead of letting them shrink as $\epsilon \to 0$.

\begin{theorem}\label{thm:ldp}
Fix $T > 0$ and $\kappa > 0$.
Mild solutions of~\eqref{eq:DSG} with initial data $u_0 \in \Holder^{1/2-\kappa}(\Tor)$
satisfy a large deviation principle on $\Holder([0,T];\, \Holder^{1/2-\kappa}(\Tor))$ with rate function
\[
I_{[0,T]}(u) = \frac 1 2 \int_0^T \int_\Tor
    \left[ \partial_t u(x,t) - \Laplace u(x,t) + \gamma \sin(\beta u(x,t)) \right]^2 \dx \dt
\]
whenever the integral is well-defined, and $I_{[0,T]}(u) = +\infty$ otherwise.
\end{theorem}
\begin{proof}
Let $\xi$ and $\xi'$ be two realizations of the noise
and $u$, $u'$ the corresponding solutions of~\eqref{eq:DSG}.
By the mild solution formula
\begin{equation}
\begin{split}
\norm{u(t) - u'(t)}_{\Holder^{1/2-\kappa}}
&\lesssim \int_0^t \norm{e^{(t-s)\Laplace} [\sin(\beta u(s)) - \sin(\beta u'(s))]}_{\Holder^{1/2-\kappa}} \diff s\\
&\quad + \sqrt{2\epsilon} \norm{\int_0^t e^{(t-s)\Laplace} \diff (\xi_s - \xi'_s)}_{\Holder^{1/2-\kappa}}.
\end{split}
\end{equation}
The first integrand can be bounded for example with
\begin{equation}
\begin{split}
&\mathrel{\phantom{=}} \norm{e^{(t-s)\Laplace} [\sin(\beta u(s)) - \sin(\beta u'(s))]}_{H^1}\\
&\leq \beta (1 + C (t-s)^{-1/2}) \norm{u(s) - u'(s)}_{L^2},
\end{split}
\end{equation}
so Grönwall's inequality gives the bound
\begin{equation}
\norm{u(t) - u'(t)}_{\Holder^{1/2-\kappa}}
\lesssim \exp(CT)
    \norm{\int_0^t e^{(t-s)\Laplace} \diff (\xi_s - \xi'_s)}_{\Holder^{1/2-\kappa}}.
\end{equation}
Since the stochastic convolution is a bounded linear map
from the parabolic Hölder space (see \cite[Definition~3.2.7]{berglund_introduction_2022})
$\Holder^{-3/2-\kappa}_{\mathfrak s}$ to $\Holder^{1/2-\kappa}_{\mathfrak s}$,
this shows that the solution $u$ is a continuous function of the noise $\xi$.

It follows from the abstract results in \cite{hairer_large_2015}
that $\Holder^{-3/2-\kappa}_{\mathfrak s}$ satisfies a large deviation principle with
rate $2\epsilon$ and rate function
\begin{equation}
I'_{[0,T]}(\xi) = \frac 1 2 \int_0^T \int_\Tor \abs{\xi(x,t)}^2 \dx \dt.
\end{equation}
Then the contraction principle \cite[Theorem~4.2.1]{dembo_large_2010} shows that
$\Holder^{1/2-\kappa}_{\mathfrak s}$ satisfies a large deviation principle with rate function
\begin{align}
I_{[0,T]}(u) &= \inf \big\{ I'_{[0,T]}(\xi) \colon
    \xi \in \Holder^{-3/2-\kappa}_{\mathfrak s},\; u = f(\xi) \big\},\\
f(\xi, t) &= e^{t\Laplace} f(\xi, t)
    - \gamma \int_0^t e^{(t-s)\Laplace} \sin(\beta f(\xi, s)) \diff s
    + \sqrt{2\epsilon} \int_0^t e^{(t-s)\Laplace} \diff \xi_s. \notag
\end{align}
When $u$ is a mild solution with the required initial data,
we can solve $\xi$ from~\eqref{eq:DSG}, which gives the stated formula.
\end{proof}

We further define for all $X \subset \Holder^{1/2-\kappa}(\Tor)$ the quantity
\begin{equation}
H(u_0, X) \coloneqq \frac 1 2 \inf_{T > 0} \inf_\varphi I_{[0,T]}(\varphi),
\end{equation}
where the second infimum is over continuous paths $\varphi \colon [0,T] \to \Holder^{1/2-\kappa}(\Tor)$
such that $\varphi(0) = u_0$ and $\varphi(t) \in X$ for some $t \leq T$.
This quantity can be estimated in terms of the relative communication height:

\begin{lemma}
For any $X \subset \Holder^{1/2-\kappa}(\Tor)$ we define
\[
\bar V(u_0, X) \coloneqq \inf_\varphi \sup_{0 \leq s \leq 1} [F(\varphi(s)) - F(u_0)],
\]
where the continuous path $\varphi \colon [0,1] \to \Holder^{1/2-\kappa}(\Tor)$
satisfies $\varphi(0) = u_0$ and $\varphi(1) \in X$.
Then $H(u_0, X) \geq \bar V(u_0, X)$.
\end{lemma}
\begin{proof}
\cite[Lemma~5.9]{berglund_sharp_2013}.
\end{proof}

These results imply the following estimate on the hitting probability
of the untruncated flow \eqref{eq:DSG}.
As soon as there is a hill to cross, the probability is exponentially decaying.
Note that the right-hand side is independent of $T$ but only holds for $\epsilon$ small enough.

\begin{lemma}\label{thm:ldp hitting prob}
Fix $T > 0$ and $\kappa' > 0$.
There exists $\epsilon_0$ dependent on $\kappa'$ and $u_0$ such that
\[
\Prob_{u_0}(\tau_X \leq T) \leq \exp\!\left(-\frac{\bar V(u_0, X) - \kappa'}{\epsilon}\right)
\quad\text{for all } \epsilon < \epsilon_0.
\]
\end{lemma}
\begin{proof}
As the noise is of order $\sqrt{2\epsilon}$, the large-deviation principle gives
\begin{equation}
\limsup_{\epsilon \to 0} 2\epsilon
    \log \Prob_{u_0}(\{ u \colon \tau_X \leq T \})
\leq -\inf_{v \in \bar\Gamma} I_{[0,T]}(v),
\end{equation}
where $\bar\Gamma$ is the closure of the random solution set $\{ u \colon \tau_X \leq T \}$.
We then use the relations
\begin{equation}
\inf_{v \in \bar\Gamma} I_{[0,T]}(v)
\geq \inf_{t > 0} \inf_{\substack{v(s) \in X\\ s \leq t}} I_{[0,t]}(v)
= 2 H(u_0, X)
\geq 2 \bar V(u_0, X),
\end{equation}
and the definition of $\limsup$ to find $\epsilon_0$ such that the statement holds.
\end{proof}

There is also the following estimate on the deviation from the deterministic path.
We use it to follow the solution as it descends into a well.

\begin{lemma}\label{thm:ldp deterministic}
Fix $T > 0$ and $\delta > 0$.
Let $\psi$ be the deterministic solution to \eqref{eq:DSG}
started from $u_0 \in \Holder^{1/2-\kappa}(\Tor)$.
Then
\[
\Prob_{u_0}\!\left( \sup_{0 \leq t \leq T} \norm{u(t) - \psi(t)}_{\Holder^{1/2-\kappa}(\Tor)} > \delta \right)
\leq \exp\!\left( -\frac{K}{\epsilon} \right)
\]
for some $K$ dependent on $\delta$, $u_0$, and $T$.
\end{lemma}
\begin{proof}
\cite[Corollary~6.6]{faris_large_1982}.
\end{proof}

With these tools we can prove the main result of this appendix.
Recall that the equilibrium potential of the $N$-truncated flow is
\begin{equation}
h_{B,A}^{(N)} \Prob\!\big(\tau_B^{(N)} < \tau_A^{(N)}\big),
\end{equation}
where
\begin{equation}
\begin{gathered}
B \coloneqq \{ \norm{u}_{\Holder^{1/2-\kappa}} < \delta \}
\text{ and }\\
A \coloneqq \{ \norm{u + 2\pi/\beta}_{\Holder^{1/2-\kappa}} < \delta \}
    \cup \{ \norm{u - 2\pi/\beta}_{\Holder^{1/2-\kappa}} < \delta \}.
\end{gathered}
\end{equation}
We emphasize the difference of truncated and non-truncated flows with superscripts in the following.
The next proof is from \cite[Proposition~5.17]{berglund_sharp_2013},
but we repeat it here for self-containedness.

In \cite{berglund_sharp_2013} the estimate is given with an $\bigO(e^{-C/\epsilon})$ upper bound,
but we state a linear bound in order to have a polynomial dependence between $\epsilon$ and $N$.
Any functional dependency is of course possible,
but the error in the capacity bound is nevertheless dominated by the potential estimate.
We need both $\epsilon$ and $N$ to control the Galerkin approximation.

\begin{proposition}\label{thm:ldp equilibrium local}
Assume that $u_0 \in \Holder^{1/2-\kappa}(\Tor)$ is in the domain of attraction for $A$ in a way that
$\bar V(u_0, B) \geq \eta$.
Then there exist $\epsilon_0$ and $N_0$,
both dependent on $u_0$, such that
\[
h_{B,A}^{(N)}(u_0) \leq 4 \epsilon,
\]
whenever $\epsilon < \epsilon_0(u_0)$ and $N > N_0(u_0) \epsilon^{-3}$.
\end{proposition}
\begin{proof}
It is no loss in generality to assume that $\hat u_0(0) > \pi/\beta$
so that we can replace $A$ with just one ball.
The idea is to approximate the truncated flow with the full flow,
for which we have large deviation estimates.

Let us define the slightly larger ball
$B_+ \coloneqq \{ w \colon \norm{w}_{\Holder^{1/2-\kappa}} < \delta + 2\delta' \}$.
We choose $\delta' \in (0, \delta/3)$ small enough that $\bar V(u_0, B_+) \geq \eta$ still.
Then $u^{(N)}$ is guaranteed to hit $B_+$ if $u$ hits $B$
and $\|{u^{(N)} - u}\|_{\Holder^{1/2-\kappa}} < \delta'$.
To model the latter requirement, we set the event
\begin{equation}
\Omega_N \coloneqq \left\{ \sup_{0 \leq t \leq T} \norm{u^{(N)}(t) - u(t)}_{\Holder^{1/2-\kappa}}
    \leq \delta' \right\},
\end{equation}
where $T > 0$ will be fixed below.

Now the equilibrium potential can be approximated as
\begin{equation}
h_{B,A}^{(N)}(u_0) \leq \Prob_{u_0}(\tau_B \leq T) + \Prob_{u_0}(\Omega_N^\complement)
    + \Prob_{u_0}\!\big(\tau_A^{(N)} > T\big).
\end{equation}
For the first term, \Cref{thm:ldp hitting prob} implies
\begin{equation}
\Prob_{u_0}(\tau_B \leq T) \leq \exp(-(\eta - \kappa)/\epsilon),
\end{equation}
once $\epsilon < \epsilon_0$ dependent on $u_0$.

Let us then define the smaller ball
$A_- \coloneqq \{ w \colon \norm{w - 2\pi/\beta}_{\Holder^{1/2-\kappa}} < \delta - 3\delta' \}$.
We choose $T$ so that the deterministic solution $\psi$ started from $u_0$ hits $A_-$ by the time $T$.
Then
\begin{equation}
\begin{split}
\Prob_{u_0}(\tau_A > T)
&\leq \Prob_{u_0}\!\left( \sup_{0 \leq t \leq T} \norm{u(t) - \psi(t)}_{\Holder^{1/2-\kappa}} > \delta' \right)
    + \Prob_{u_0}(\Omega_N^\complement)\\
&\leq \exp\!\left(\! -\frac{K}{\epsilon} \right) + \Prob_{u_0}(\Omega_N^\complement)
\end{split}
\end{equation}
by \Cref{thm:ldp deterministic}.
Finally, \Cref{thm:galerkin approximation} implies that
\begin{equation}
\Prob_{u_0}(\Omega_N^\complement)
\leq \frac{C(T) N^{-\kappa/2}}{\delta'}.
\end{equation}
By choosing $N$ large enough, we can make this term smaller than $\epsilon$.
\end{proof}

\begin{corollary}\label{thm:ldp equilibrium global}
If all $u_0$ in a compact subset $D \subset \Holder^{1/2-\kappa}(\Tor)$
satisfy the assumptions of \Cref{thm:ldp equilibrium local},
then the parameters $\epsilon_0$ and $N_0$ can be chosen uniformly in $u_0 \in D$.
\end{corollary}
\begin{proof}
By compactness it suffices to show that every $u_0 \in D$ has a neighbourhood
where \Cref{thm:ldp equilibrium local} holds
with parameters $\epsilon_0(u_0)$, $N_0(u_0)$ and $T(u_0)$.
Then finitely many neighbourhoods cover $D$,
and we can take minima/maxima of these parameters.

Fix $u_0$, and let $w_0 \in \projN B_{\Holder^{1/2-\kappa}}(\projN u_0, \rho) \cap D$.
We denote by $w^{(N)}$ the solution to \eqref{eq:DSG truncated} with data $w_0$.
It then suffices to choose $\rho$ so that
\begin{equation}
\sup_{0 \leq t \leq T} \norm{u^{(N)}(t) - w^{(N)}(t)}_{\Holder^{1/2-\kappa}} < \delta',
\end{equation}
since the definitions of sets $A_-$ and $B_+$ in \Cref{thm:ldp equilibrium local} leave some wiggle room.

For the rest of this proof we leave out the $(N)$ superscripts for brevity.
When the noise is fixed, the mild solution formula and Besov embedding give
\begin{equation}
\begin{split}
&\mathrel{\phantom{=}} \norm{u(t) - w(t)}_{\Holder^{1/2-\kappa}(\Tor)}\\
&= \norm{\projN e^{t \Laplace} [u_0 - w_0]
    - \gamma \int_0^t \projN e^{(t-s) \Laplace} [\sin(\beta u(s)) - \sin(\beta w(s))] \diff s}_{\Holder^{1/2-\kappa}}\\
&\lesssim \norm{e^{t \Laplace} \projN [u_0 - w_0]}_{\Holder^{1/2-\kappa}}
    + \int_0^t \norm{e^{(t-s) \Laplace} [\sin(\beta u(s)) - \sin(\beta w(s))]}_{H^1} \diff s.
\end{split}
\end{equation}
For the first term we use uniform-in-$t$ boundedness of the heat kernel and
the assumption $\norm{\projN [u_0 - w_0]}_{\Holder^{1/2-\kappa}} \leq \rho$.

As in \Cref{thm:ldp} we use the smoothing property of the heat semigroup
to estimate the second term by
\begin{equation}
C \int_0^t (1 + (t-s)^{-1/2}) \norm{\sin(\beta u(s)) - \sin(\beta w(s))}_{L^2} \diff s.
\end{equation}
The integrand can be estimated with Lipschitz continuity.
Then Grönwall's inequality implies
\begin{equation}
\begin{split}
\norm{u(t) - w(t)}_{\Holder^{1/2-\kappa}(\Tor)}
&\lesssim \rho \exp\!\left( C \int_0^t 1 + (t-s)^{-1/2} \diff s \right)\\
&\lesssim \rho \exp(CT).
\end{split}
\end{equation}
We now choose $\rho$ such that this term is smaller than $\delta'$.
Here $T$ will depend on $u_0$, but the bound is uniform in $N$.
\end{proof}

\raggedright
\printbibliography

\end{document}